\documentclass[paper=A4,pagesize]{scrartcl}

\usepackage[english]{babel}

\usepackage{amsmath,amssymb,amsthm,amsfonts}
\usepackage{exscale}

\usepackage{bbm}
\usepackage{lmodern}

\usepackage{hyperref}

\usepackage{slashed}

\usepackage{graphicx}

\usepackage[all]{xy}
\CompileMatrices
\usepackage{tensor}

\usepackage{ifthen}

\numberwithin{equation}{section}

\newcommand{\K}{
\mathbb{K}
}

\newcommand{\N}{
\mathbb{N}
}

\newcommand{\R}{
\mathbb{R}
}
\newcommand{\C}{
\mathbb{C}
}
\newcommand{\RP}{
\mathbb{RP}
}

\newcommand{\tp}{
\otimes
}

\newlength{\wurelwidth}
\newcommand{\urel}[2][=]{\mathrel{\mathop{#1}\limits_{\!\scalebox{0.5}{#2}\!}}}
\newcommand{\wurel}[2][=]{\mathrel{\mathop{#1}_{\!\scalebox{0.5}{\makebox[\the\wurelwidth]{#2}}\!}}}

\newcommand{\comps}{\pi_0}
\newcommand{\Edges}[1]{E(#1)}
\newcommand{\dimension}{D}
\newcommand{\Graph}[2][1.0]{%
\vcenter{\hbox{\includegraphics[scale=#1]{#2}}}%
}

\newcommand{\phipsi}{%
{\frac{\phipol}{\psipol}}%
}
\newcommand{\phipol}{\varphi}
\newcommand{\psipol}{\psi}

\newcommand{\subdiv}{\prec}
\newcommand{\subdiveq}{\preceq}

\newcommand{\supdiveq}{\succeq}

\newcommand{\nosubdiveq}{\npreceq}

\newcommand{\nosupdiveq}{\nsucceq}

\newcommand{\coeff}[2][]{c_{#2}^{#1}}
\newcommand{\rp}{\mu} % Subtraction point for momentum scheme

\newcommand{\PhiR}{\Phi_+}
\newcommand{\dimPhi}{{{}_{\reg}\Phi}}

\newcommand{\toyform}{\eta}
\newcommand{\toycc}{L}
\newcommand{\scalelog}{\ell}
\newcommand{\x}{x}
\newcommand{\period}{\mathcal{P}}
\newcommand{\reg}{z} % Parameter fuer dimreg und analytische Regularisierung
\newcommand{\toy}[1][\reg]{%
{_{#1}}\phi
}% Feynmanregeln fuer das Toymodel
\newcommand{\toyR}[1][]{%
	\def\ArgI{#1}
	\toyRRelay
}
\newcommand\toyRRelay[1][\reg]{%
{^{}_{#1}}\phi^{}_{R\ifthenelse{\equal{\ArgI}{}}{}{,\ArgI}}
}% Renormierte Feynmanregeln fuer das Toymodel
\newcommand{\MStoyR}[1][]{%
{^{}_{\reg}}\phi^{}_{\text{\tiny MS}\ifthenelse{\equal{#1}{}}{}{,#1}}
}% Renormierte Feynmanregeln fuer das Toymodel
\newcommand{\toyphy}[1][]{{}^{}_0\phi^{}_{#1}}
\newcommand{\MStoyphy}[1][]{{}^{}_{0}\phi^{}_{\text{\tiny MS}\ifthenelse{\equal{#1}{}}{}{,#1}}}
\newcommand{\toyZ}{Z} % Counterterms fuer das Toymodel
\newcommand{\toylog}{\gamma} %beta-function bzw. anomalous dimension

\newcommand{\intrules}[1][]{{^{#1}}\varphi}
\newcommand{\hide}[1]{}
\newcommand{\Rms}{R_{\text{\tiny MS}}} % Minimal Subtraction Scheme
\newcommand{\bigo}[1]{\mathcal{O}\left(#1\right)}
\newcommand{\momsch}[1]{R_{#1}}
\newcommand{\dilat}{\Lambda}
\newcommand{\coupling}{g}
\newcommand{\powdep}{\kappa}

\newcommand{\forests}{\mathcal{F}}
\newcommand{\f}{w}
\newcommand{\autoconc}{\triangleright} % Produkt auf H_R' induziert von der Konkatenation in Aut(H_R) ueber die universelle Eigenschaft
\newcommand{\unimor}[1]{{^{#1}}\!\rho} % Morphismus induziert von der universellen Eigenschaft von H_R
\newcommand{\trees}{\mathcal{T}}
\newcommand{\tree}[1]{%
\vcenter{\hbox{\includegraphics{tree#1}}}%
}

\newcommand{\convolution}{\star}
\newcommand{\counit}{\varepsilon}
\newcommand{\unit}{u}
\newcommand{\Prim}{
\mathrm{Prim}
}
\newcommand{\polyint}{
\int_0
}
\newcommand{\asymptotic}{
\sim
}
\newcommand{\cored}{\widetilde{\Delta}} % Koproduct ohne 1 tensor id und id tensor 1

\newcommand{\chars}[2]{{G}_{#2}^{#1}} % Gruppe der Charaktere (Algebrenmorphismen) unter Konvolution
\newcommand{\infchars}[2]{\mathfrak{g}_{#2}^{#1}} % Gruppe der Charaktere (Algebrenmorphismen) unter Konvolution
\newcommand{\auto}[1]{{^{#1}}\chi} % Automorphismen von H_R aus Funktionalen mittels universeller Eigenschaft
\newcommand{\alg}[1][A]{\mathcal{#1}}
\newcommand{\decor}{\mathcal{D}}
\newcommand{\gradAut}{\theta}
\newcommand{\dH}{\delta} % boundary map
\newcommand{\HH}[1][]{\text{HH}^{#1}_{\counit}}
\newcommand{\HZ}[1][]{\text{HZ}^{#1}_{\counit}}
\newcommand{\HB}[1][]{\text{HB}^{#1}_{\counit}}

\newcommand{\dd}[1][]{\mathrm{d}^{#1}}
\newcommand{\restrict}[2]{%
{\left. #1 \right|}_{#2}%
}
\DeclareMathOperator{\Res}{Res}
\DeclareMathOperator{\Hom}{Hom}
\DeclareMathOperator{\lin}{lin}
\DeclareMathOperator{\Aut}{Aut}
\DeclareMathOperator{\End}{End}
\newcommand{\ev}{
\mathrm{ev}
}
\DeclareMathOperator{\im}{im}
\renewcommand{\1}{
\mathbbm{1}
}
\newcommand{\defas}{
\mathrel{\mathop:}=
}

\newcommand{\set}[1]{
\left\{ #1 \right\}
}
\newcommand{\setexp}[2]{
\left\{ #1\!:\ #2 \right\}
}
\newcommand{\abs}[1]{
\left\lvert #1 \right\rvert
}
\newcommand{\id}{
\mathrm{id}
}

\newcommand{\qft}{%
quantum field theory%
}
\newcommand{\qfts}{%
quantum field theories%
}
\newcommand{\momscheme}{kinetic scheme}

\newtheorem{satz}{Theorem}[section]
\newtheorem{definition}[satz]{Definition}
\newtheorem{lemma}[satz]{Lemma}
\newtheorem{korollar}[satz]{Corollary}
\newtheorem{proposition}[satz]{Proposition}
\newtheorem{beispiel}[satz]{Example}

\title{\vspace{-1cm}Renormalization and Mellin transforms}
\author{
Dirk Kreimer\thanks{Alexander von Humboldt Chair in Mathematical Physics, supported by the Alexander von Humboldt Foundation and the BMBF.}
\and
Erik Panzer\thanks{\href{mailto:panzer@mathematik.hu-berlin.de}{\nolinkurl{panzer@mathematik.hu-berlin.de}}}
}
\date{\small
Institutes of Physics and Mathematics, Humboldt-Universit\"{a}t zu Berlin, \\
Unter den Linden 6, 10099 Berlin, Germany
\vspace{-1cm}%
}

\begin{document}
\maketitle
\begin{abstract}\small
	We study renormalization in a {\momscheme} using the Hopf algebraic framework, first summarizing and recovering known results in this setting. Then we give a direct combinatorial description of renormalized amplitudes in terms of Mellin transform coefficients, featuring the universal property of rooted trees $H_R$. In particular, a special class of automorphisms of $H_R$ emerges from the action of changing Mellin transforms on the Hochschild cohomology of perturbation series.

	Furthermore, we show how the Hopf algebra of polynomials carries a refined renormalization group property, implying its coarser form on the level of correlation functions. Application to scalar {\qft} reveals the scaling behaviour of individual Feynman graphs.
\end{abstract}

\section{Introduction}

As was shown in \cite{Kreimer:HopfAlgebraQFT,Yeats,BrownKreimer:AnglesScales}, we may decompose Feynman integrals into functions of a single \emph{scale} parameter $s$ only (further forking into logarithmic divergent parts multiplied by suitable powers of $s$) and scale-independent functions of the other kinematic variables, called \emph{angles}. Furthermore, the Hopf algebra $H_R$ of rooted trees suffices to encode the full structure of subdivergences in {\qft} by \cite{Kreimer:HopfAlgebraQFT,CK:NC,CK:RH1}.

We can therefore study such generic Feynman rules in a purely algebraic framework as pioneered in \cite{Factorization,CK:RH1}. Renormalizing short-distance singularities by subtraction at a reference scale $\rp$ (\emph{\momscheme}) leads to amplitudes of a distinguished algebraic kind: Theorem \ref{satz:toymodel-universal} proves them to implement the universal property of $H_R$, delivering an explicit combinatorial evaluation in terms of Mellin transform coefficients.

Further investigating the role of Hochschild cohomology, in section \ref{sec:H_R} we define a class of automorphisms of $H_R$  which transform the perturbation series in a way equivalent to changing the Feynman rules. This clarifies how exact one-cocycles describe variations.

In sections \ref{sec:polynomials} and $\ref{sec:DSE}$ we advertise to think about the renormalization group property as a Hopf algebra morphism to polynomials, determining higher logarithms in \eqref{eq:toylog}. We show how it implies the renormalization group on correlation functions and extend the \emph{propagator-coupling-duality} of \cite{Kreimer:ExactDSE} which yields the functional equation \eqref{eq:propagator-coupling}.

After analysing the differences to the minimal subtraction scheme in section \ref{sec:ms}, we show explicitly how our general results manifest themselves in scalar field theory.

\section{Connected Hopf algebras}

The fundamental mathematical structure behind perturbative renormalization is the Hopf algebra as discovered in \cite{Kreimer:HopfAlgebraQFT}. We briefly summarize the results on Hopf algebras we need and recommend \cite{Manchon,Panzer:Master} for detailed introductions with a focus on renormalization.

All vector spaces live over a field $\K$ of zero characteristic (in examples $\K=\R$), $\Hom(\cdot, \cdot)$ denotes $\K$-linear maps and $\lin M$ the linear span.
Every algebra $(\alg,m,u)$ shall be unital, associative and commutative, any bialgebras $(H,m,u,\Delta,\counit)$ in addition also counital and coassociative. They split into the scalars and the \emph{augmentation ideal} $\ker\counit$ as $H = \K\! \cdot\! \1 \oplus \ker \counit = \im \unit \oplus \ker \counit$, inducing the projection $P \defas \id - \unit \circ \counit\!:\ H \twoheadrightarrow \ker \counit$. We use Sweedler's notation $\Delta (x) = \sum_x x_1 \tp x_2$ and $\cored(x)=\sum_x x'\tp x''$ to abbreviate the \emph{reduced coproduct}
$
		\cored \defas \Delta - \1 \tp \id - \id \tp \1.
$

We assume a \emph{connected grading} $H = \bigoplus_{n\geq 0} H_n$ ($H_0 = \K\cdot\1$) and write $\abs{x}\defas n$ for homogeneous $0\neq x\in H_n$, defining the \emph{grading operator} $Y\in\End(H)$ by $Yx = \abs{x}\cdot x$.
Exponentiation yields a one-parameter group $\K \ni t \mapsto \gradAut_t$ of Hopf algebra automorphisms
\begin{equation}
	\gradAut_t
	\defas \exp (tY)
	= \sum_{n\in\N_0} \frac{(tY)^n}{n!}
	,\quad
	\forall n\in\N_0:\quad
	H_n \ni x \mapsto \gradAut_t(x)
	= e^{t\abs{x}} x
	= e^{nt} x.
	\label{eq:grad-aut}
\end{equation}
Given an algebra $(\alg,m_{\alg},u_{\alg})$, the associative \emph{convolution product} on $\Hom(H,\alg)$ is 
\begin{equation*}
	\Hom(H,\alg)
	\ni	\phi,\psi	\mapsto
	\phi \convolution \psi
	\defas
	m_{\alg} \circ (\phi\tp \psi) \circ \Delta
	\in \Hom(H,\alg),
\end{equation*}
with unit given by $e\defas u_{\alg} \circ \counit$. As outcome of the connectedness of $H$ we stress
\begin{enumerate}
	\item
		The \emph{characters} $\chars{H}{\alg} \defas \setexp{\phi\in\Hom(H,\alg)}{\phi\circ u=u_{\alg} \ \text{and}\ \phi\circ m = m_{\alg} \circ (\phi\tp \phi)}$ (morphisms of unital algebras) form a group under $\convolution$.

	\item
		Hence $\id\in\chars{H}{H}$ has a unique inverse $S\defas\id^{\convolution-1}$, called \emph{antipode}, turning $H$ into a Hopf algebra. For all $\phi\in\chars{H}{\alg}$ we have $\phi^{\convolution -1} = \phi \circ S$.

	\item
		The bijection $\exp_{\convolution}\!:\ \infchars{H}{\alg}\rightarrow\chars{H}{\alg}$ with inverse $\log_{\convolution}\!:\ \chars{H}{\alg}\rightarrow\infchars{H}{\alg}$ between $\chars{H}{\alg}$ and the \emph{infinitesimal characters} $\infchars{H}{\alg} \defas \setexp{\phi\in\Hom(H,\alg)}{\phi\circ m = \phi\tp e + e\tp \phi}$ is given by the pointwise finite series
	\begin{equation}\label{eq:log-exp}
		\exp_{\convolution} (\phi)
		\defas \sum_{n\in\N_0} \frac{\phi^{\convolution n}}{n!}
		\quad\text{and}\quad
		\log_{\convolution} (\phi)
		\defas \sum_{n\in\N} \frac{(-1)^{n+1}}{n} (\phi-e)^{\convolution n}.
	\end{equation}
\end{enumerate}

\subsection{Hochschild cohomology}

The Hochschild cochain complex \cite{CK:NC,BergbauerKreimer,Panzer:Master} we associate to $H$ contains the functionals $H' = \Hom(H,\K)$ as zero-cochains.
One-cocycles $L \in \HZ[1](H) \subset \End(H)$ are linear maps such that $\Delta \circ L = (\id \tp L) \circ \Delta + L \tp \1$ and the differential
\begin{equation}\label{eq:dH}
	\dH: H'\rightarrow \HZ[1](H),
	\alpha \mapsto
	\dH\alpha \defas (\id\tp\alpha) \circ \Delta - u \circ \alpha
	\in \HB[1](H)\defas \dH\left( H' \right)
\end{equation}
determines the first cohomology group by $\HH[1](H) \defas \HZ[1](H) / \HB[1](H)$.

\begin{lemma}\label{satz:cocycle-props}
	Cocycles $L\in \HZ[1](H)$ fulfil $\im L \subseteq \ker \counit$ and $L(\1) \in \Prim(H)\defas\ker\cored$ is primitive. The map $\HH[1](H) \rightarrow \Prim(H)$, $[L] \mapsto L(\1)$ is well-defined since $\dH\alpha(\1)=0$ for all $\alpha\in H'$.
\end{lemma}

\subsection{Rooted Trees}

The Hopf algebra $H_R$ of rooted trees serves as the domain of Feynman rules. As an algebra, $H_R=S(\lin\trees) = \K[\trees]$ is free commutative\footnote{
We consider \emph{unordered} trees $\scalebox{0.7}{$\tree{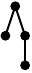}$}=\scalebox{0.7}{$\tree{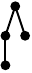}$}$ and forests $\tree{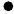}\tree{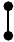}=\tree{++--}\tree{+-}$, sometimes called \emph{non-planar}.
}
 generated by the \emph{rooted trees} $\trees$ and spanned by their disjoint unions (products) called \emph{rooted forests} $\forests$:
	\begin{equation*}
		\trees = 
			\set{\tree{+-}, \tree{++--}, \tree{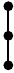}, \tree{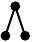}, \tree{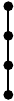}, \tree{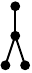}, \tree{++-++---}, \tree{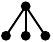}, \ldots}
		,\quad
		\forests = 
				\set{\1} 
	\cup 	\trees
	\cup 	\set{\tree{+-}\tree{+-}, \tree{+-}\tree{+-}\tree{+-}, \tree{+-}\tree{++--}, \tree{+-}\tree{+-}\tree{+-}\tree{+-}, \tree{+-}\tree{+-}\tree{++--}, \tree{+-}\tree{++-+--}, \tree{+-}\tree{+++---}, \ldots}.
	\end{equation*}
Every $\f\in\forests$ is just the monomial $\f=\prod_{t\in\comps(\f)} t$ of its multiset of tree components $\comps(\f)$, while $\1$ denotes the empty forest. The number $\abs{\f}\defas \abs{V(\f)}$ of nodes $V(\f)$ induces the grading $ H_{R,n} = \lin \forests_n $ where $\forests_n \defas \setexp{\f\in\forests}{\abs{\f}  = n}$.
\begin{definition}
	The (linear) \emph{grafting operator} $B_+ \in \End(H_R)$ attaches all trees of a forest to a new root, so for example
$	B_+ \left( \1 \right)	= \tree{+-} $,
$	B_+ \left( \tree{+-} \right) = \tree{++--}$ and
$	B_+ \left( \tree{+-}\tree{+-} \right)	= \tree{++-+--}$.
\end{definition}
Clearly, $B_+$ is homogenous of degree one with respect to the grading and restricts to a bijection $B_+\!:\ \forests \rightarrow \trees$. The coproduct $\Delta$ is defined to make $B_+$ a cocycle by requiring
\begin{equation}
	\Delta \circ B_+
	= B_+ \tp \1 + (\id \tp B_+) \circ \Delta.
	\label{eq:B_+-cocycle}
\end{equation}
\begin{lemma}\label{satz:B_+-cocycle}
	In cohomology, $0\neq [B_+] \in \HH[1](H_R)$ is non-trivial by $B_+(\1) = \tree{+-} \neq 0$.
\end{lemma}
It characterizes $H_R$ through the well-known (theorem 2 of \cite{CK:NC}) \emph{universal property} of
\begin{satz}\label{satz:H_R-universal}
	To an algebra $\alg$ and $L \in \End(\alg)$ there exists a unique morphism $\unimor{L}\!: H_R \rightarrow \alg$ of unital algebras such that
	\begin{equation}
		\unimor{L} \circ B_+ = L \circ \unimor{L},
		\quad \text{equivalently} \quad
		\vcenter{\xymatrix{
		{H_R} \ar[r]^{\unimor{L}} \ar[d]_{B_+} & {\alg} \ar[d]^{L} \\
			{H_R} \ar[r]_{\unimor{L}} & {\alg}
		}}
		\quad \text{commutes.}
		\label{eq:H_R-universal}
	\end{equation}
	In case of a bialgebra $\alg$ and a cocycle $L \in \HZ[1](\alg)$, $\unimor{L}$ is a morphism of bialgebras and even of Hopf algebras when $\alg$ is Hopf.
\end{satz}
This morphism $\unimor{L}$ simply replaces $B_+$, $m_{H_R}$ and $\1$ as placeholders by $L$, $m_{\alg}$ and $\1_{\alg}$:
\begin{equation*}
	\unimor{L} \left( \tree{++-+--} - 3\tree{+-} \right)
	= \unimor{L} \left\{ B_+ \left( {\left[B_+(\1) \right]}^2 \right) - 3 B_+ (\1) \right\}
	= L \left(  {\left[ L(\1_{\alg}) \right]}^2 \right)  - 3 L(\1_{\alg}).
\end{equation*}
\begin{beispiel}\label{ex:tree-factorial}
	The cocycle $\polyint\in\HZ[1](\K[\x])$ of section \ref{sec:polynomials} induces the character
	\begin{equation}\label{eq:int-rules}
		\intrules \defas \unimor{\polyint}\in\chars{H_R}{\K[\x]}
		\quad\text{fulfilling}\quad
		\intrules (\f) 
		= \frac{\x^{\abs{\f}}}{\f!}
		\quad\text{for any forest}\quad
		\f\in\forests,
		\quad\text{using}
	\end{equation}
\end{beispiel}
\hide{
\begin{proof}
	The inductive proof (start at $f=\1$ trivial) supposes \eqref{eq:int-rules} to be true for all $f\in \forests_{\leq n}$. Then \eqref{eq:int-rules} also holds for any forest $f\in\forests_{n+1}$ with $\abs{\pi_0 (f)} > 1$ as
	\begin{equation*}
		  \intrules (f) 
		= \prod_{t \in \pi_0 (f)} \intrules (t) 
		\urel{\eqref{eq:int-rules}}
			\prod_{t \in \pi_0 (f)} \frac{\x^{\abs{t}}}{t!} 
		= \frac{\x^{\sum_{t\in \pi_0(f)} \abs{t}}}{\prod_{t \in \pi_0(f)} t!} 
		\urel{\eqref{eq:tree-factorial}}
			\frac{\x^{\abs{f}}}{f!},
	\end{equation*}
	exploiting $\abs{t} \leq n$ for any $t\in\pi_0(f)$ to use the induction hypothesis. It remains to consider a tree $t = B_+ (f)$ for some $f\in \forests_n$ in
	\begin{equation*}
		\intrules (t)
		= \intrules \circ B_+ (f) 
		\urel{\eqref{eq:H_R-universal}}
			\int_0 \circ \:\intrules (f) 
		\urel{\eqref{eq:int-rules}}
			\int_0^\x \frac{y^{\abs{f}}}{f!} \ \dd y
		= \frac{\x^{\abs{f}+1}}{(\abs{f}+1) \cdot f!} 
		= \frac{\x^{\abs{B_+ (f)}}}{\left( B_+ f \right)!}
		\urel{\eqref{eq:tree-factorial}}
			\frac{\x^{\abs{t}}}{t!}. \qedhere
	\end{equation*}
\end{proof}
}%
\begin{definition}\label{def:tree-factorial}
	The \emph{tree factorial} $(\cdot)!\in\chars{H_R}{\K}$ is equivalently determined by requesting
	\begin{equation}\label{eq:tree-factorial}
		\left[ B_+(\f) \right]! 
		= \f! \cdot \abs{B_+(\f)}
		\quad\text{or}\quad
		\f! 
		\urel{\footnotemark}
		\prod_{v\in V(\f)} \abs{\f_v}
		\quad\text{for all}\quad
		\f\in\forests.
	\end{equation}
	\footnotetext{By $\f_v$ we denote the subtree of $\f$ rooted at the node $v\in V(\f)$.}
\end{definition}

\section{The generic model}
\label{sec:toymodel}

As explained in the introduction we consider Feynman rules as characters $\phi\in\chars{H_R}{\alg}$, mapping a rooted tree to a function of the parameter $s$ (by proposition \ref{prop:toymodel-mellin} it lies in the algebra $\alg	= \K[\reg^{-1},\reg]] [s^{-\reg}]$). Since $B_+$ mimics the insertion of a subdivergence into a fixed graph $\gamma$ (restricting to a single insertion place by a result from \cite{Yeats}), applying $\phi$ yields a subintegral and therefore
\begin{definition}\label{def:toymodel}
	The generic Feynman rules $\toy$ are given through theorem \ref{satz:H_R-universal} by 
	\begin{equation}\label{eq:toymodel}
		\toy_s \circ B_+ 
		= \int_0^{\infty} \frac{f(\frac{\zeta}{s})\zeta^{-\reg}}{s}\ \toy_{\zeta} \ \dd\zeta 
		= \int_0^{\infty} f(\zeta)(s\zeta)^{-\reg}\ \toy_{s\zeta} \ \dd\zeta.
	\end{equation}
\end{definition}
The integration kernel $f$ is specified by $\gamma$ after \emph{Wick rotation} to Euclidean space, with the asymptotic behaviour $f(\zeta) \sim \zeta^{-1}$ for $\zeta \rightarrow \infty$ generating the (logarithmic) divergences of these integrals (we do not address infrared problems and exclude any poles in $f$).
The regulator $\zeta^{-\reg}$ ensures convergence when $0<\Re(\reg)<1$, with results depending analytically on $\reg$. We can perform all the integrals using this \emph{Mellin transform}
\begin{equation}\label{eq:mellin-trafo}
	F(\reg) \defas
	\int_0^{\infty} f(\zeta) \zeta^{-\reg} \ \dd\zeta
	= \sum_{n=-1}^{\infty} \coeff{n} {\reg}^n,
	\quad\text{by}
\end{equation}
\begin{proposition}\label{prop:toymodel-mellin}
	For any forest $\f \in \mathcal{F}$ we have (called \emph{BPHZ model} in \cite{BroadhurstKreimer:Auto})
	\begin{equation}\label{eq:toymodel-mellin}
		\toy_{s} (\f)
		= s^{-\reg\abs{\f}} \prod_{v \in V(\f)} F \left( \reg \abs{\f_v} \right).
	\end{equation}
\end{proposition}

\begin{proof}
	As both sides of \eqref{eq:toymodel-mellin} are clearly multiplicative, it is enough to prove the claim inductively for trees. Let it be valid for some forest $\f\in\forests$, then for $t = B_+ (\f)$ observe
	\begin{align*}
		\toy_{s} \circ B_+ (\f)
		&= \int_0^{\infty} (s\zeta)^{-\reg} f(\zeta)\ \toy_{s\zeta} (\f) \ \dd\zeta
		 = \int_0^{\infty} (s\zeta)^{-\reg} f(\zeta) (s\zeta)^{-\reg\abs{\f}} \prod_{v \in V(\f)} F \left( z \abs{\f_v} \right) \ \dd\zeta \\
		&= s^{-\reg \abs{B_+ (\f)}} \left[ \prod_{v \in V(\f)} F \left( \reg \abs{\f_v} \right) \right] F \left( \reg \abs{B_+ (\f)} \right)
		= s^{-\reg \abs{t}} \prod_{v \in V(t)} F \left( \reg \abs{ t_v } \right) \qedhere.
	\end{align*}
\end{proof}
\begin{beispiel}
	Using \eqref{eq:toymodel-mellin}, we can directly write down the Feynman rules like
\begin{equation*}
	\toy_s \left( \tree{+-} \right)
	= s^{-\reg} F(\reg),
	\quad
	\toy_s \left( \tree{++--} \right)
	= s^{-2\reg} F(\reg) F(2\reg)
	\quad\text{and}\quad
	\toy_s \left( \tree{++-+--} \right)
	= s^{-3\reg} {\left[ F(\reg) \right]}^2 F(3\reg).
\end{equation*}
\end{beispiel}
Many examples (choices of $F$) are discussed in \cite{BroadhurstKreimer:Auto}, the particular case of the one-loop propagator graph $\gamma$ of Yukawa theory is in \cite{Kreimer:ExactDSE} and for scalar Yukawa theory in six dimensions one has
$
	F(\reg)
	=
	\frac{1}{\reg(1-\reg)(2-\reg)(3-\reg)}
$
as in \cite{Panzer:Master}. Already noted in \cite{Kreimer:ChenII}, the highest order pole of $\toy_s (\f)$ is independent of $s$ and just the tree factorial
	\begin{equation}\label{eq:toymodel-leading-pole}
		\toy_s (\f)
		\in s^{-\reg\abs{\f}} \hspace{-1mm}
			\prod_{v \in V(\f)} \left\{
				\tfrac{\coeff{-1}}{\reg\abs{\f_v}} 
				+ \K[[\reg]]
			\right\} 
		\urel[\subset]{\eqref{eq:tree-factorial}}
		\frac{1}{\f!} {\left( \tfrac{\coeff{-1}}{\reg} \right)}^{\abs{\f}}
				+ \reg^{1-\abs{\f}} {\K}[\ln s][[\reg]].
	\end{equation}

\subsection{Renormalization}
\label{sec:renormalization}

Algebraically, renormalization equals a \emph{Birkhoff decomposition} \cite{CK:RH1,Manchon,Panzer:Master} of $\phi\in\chars{H}{\alg}$ into \emph{renormalized} rules $\phi_R \defas \phi_+\in\chars{H}{\alg}$ and the \emph{counterterms} $Z \defas \phi_-\in\chars{H}{\alg}$ such that
\begin{equation}
	\phi
	= \phi_-^{\convolution -1} \convolution \phi_+
	\quad\text{and}\quad
	\phi_{\pm} \left( \ker\counit \right) \subseteq \alg_{\pm},
	\label{eq:birkhoff}
\end{equation}
with respect to a splitting \mbox{$\alg = \alg_+ \oplus \alg_-$} determined by the \emph{renormalization scheme} (the projection $R\!: \alg \twoheadrightarrow \alg_-$). We comment on \emph{minimal subtraction} in \ref{sec:ms} and now focus on
\begin{definition}
	On the target algebra $\alg$ of regularized Feynman rules depending on a single external variable $s$, define the \emph{\momscheme} by evaluation at $s=\rp$:
	\begin{equation}
		\End(\alg) \ni \momsch{\rp} \defas \ev_{\rp} = \left( \alg \ni f \mapsto {\left. f \right|}_{s=\rp} \right).
		\label{eq:momentum-scheme}
	\end{equation}
\end{definition}
This scheme exploits that subtraction improves the decay at infinity: Let $f(\zeta) \asymptotic \frac{1}{\zeta}$, meaning $f(\zeta) = \frac{1}{\zeta} + \tilde{f}(\zeta)$ for some $\tilde{f}(\zeta) \in \bigo{\zeta^{-1-\varepsilon}}$ with $\varepsilon>0$. Then $\toy_s(\tree{+-})$ is logarithmically divergent (would it not be for the regulator $\zeta^{-\reg}$), but subtraction
\begin{equation}
	\toy_s(\tree{+-}) - \toy_{\rp}(\tree{+-})
	= \int_0^{\infty} \left[ \frac{f(\frac{\zeta}{s})}{s} - \frac{f(\frac{\zeta}{\rp})}{\rp} \right] \zeta^{-\reg}
	= \int_0^{\infty} \left[ \frac{\tilde{f}(\frac{\zeta}{s})}{s} - \frac{\tilde{f}(\frac{\zeta}{\rp})}{\rp} \right] \zeta^{-\reg}
	\label{eq:momsch-example}
\end{equation}
yields a convergent integral even for $\reg=0$. As $\momsch{\rp}$ is a character of $\alg$, the Birkhoff recursion simplifies to \mbox{$\toyZ = \momsch{\rp} \circ \toy \circ S = \toy_{\rp} \circ S$} and $\toyR = \toy_{\rp}^{\convolution-1} \convolution \toy_{s}$.
\begin{beispiel}
	We find
	$
		\toyR[s] \left( \tree{+-} \right) 
		= \left( s^{-\reg} - \rp^{-\reg} \right) F(\reg)
	$
	and 
	$
		S\left( \tree{++--} \right) 
		= -\tree{++--} + \tree{+-}\tree{+-}
	$
	results in
\begin{equation}
	\toyR[s] \left( \tree{++--} \right)
		= \left( s^{-2\reg} - \rp^{-2\reg} \right) F(\reg)F(2\reg) - \left( s^{-\reg} - \rp^{-\reg} \right) \rp^{-\reg} F^2(\reg).
	\label{toyR:(())}
\end{equation}
\end{beispiel}

The goal of renormalization is to assure the \emph{finiteness} of the \emph{physical limit}
\begin{equation}
	\toyR[][0] \defas \lim_{\reg \rightarrow 0} \toyR,
	\label{eq:toymodel-physical}
\end{equation}
and indeed we find the finite $\toyR[s][0] \left( \tree{+-} \right) = - \coeff{-1} \ln \tfrac{s}{\rp}$. In the case of \eqref{toyR:(())} check
\begin{align}
	&\toyR[s][0] \left( \tree{++--} \right)
		= \lim_{\reg \rightarrow 0} \left\{
		- \left[ -\reg \ln \tfrac{s}{\rp} + \tfrac{\reg^2}{2} \left( \ln^2 s + 2\ln s \ln\rp - 3\ln^2 \rp \right) \right] \cdot \left[ \tfrac{\coeff[2]{-1}}{\reg^2} + 2\tfrac{\coeff{-1} \coeff{0}}{\reg} \right] \right. \nonumber\\
	&	\quad + \left. \left[ -2\reg \ln \tfrac{s}{\rp} + 2\reg^2\left( \ln^2 s - \ln^2 \rp \right)\right]\cdot \left[ \tfrac{\coeff[2]{-1}}{2\reg^2} + \tfrac{3\coeff{0} \coeff{-1}}{2\reg} \right] \right\}
		= \frac{\coeff[2]{-1}}{2} \ln^2 \tfrac{s}{\rp} - \coeff{-1} \coeff{0} \ln \tfrac{s}{\rp},
		\label{toyR-physical:(())}
\end{align}
where all poles in $\reg$ perfectly cancel. Note that $\toyR[s][0]$ maps a forest $w$ to a polynomial in $\K[\ln \tfrac{s}{\rp}]$ of degree $\leq\!\abs{w}$ without constant term (except for $\toyphy(\1) = 1$), due to the subtraction at $s=\rp$. We now prove these properties in general, extending work in \cite{Factorization}.

\subsection{Subdivergences}
\label{sec:subdivergences}

Inductively, the Birkhoff decomposition is constructed as $\phi_+(x) = (\id-\momsch{\rp}) \bar{\phi}(x)$ where
\begin{equation*}
	\bar{\phi} (x) \defas
	\phi(x) + \sum_x \phi_- (x') \phi(x'') 
	= \phi(x) + [ \phi_- \convolution \phi - \phi_- - \phi ](x) 
	= \phi_+(x) - \phi_-(x)
\end{equation*}
is the \emph{Bogoliubov character} (\emph{$\bar{R}$-operation}) and renormalizes the \emph{subdivergences}. Note
\begin{satz}\label{satz:subdivergences}
	For an endomorphism $L \in \End(\alg)$ consider the Feynman rules $\phi \defas \unimor{L}$ induced by \eqref{eq:H_R-universal}. Given a renormalization scheme $R \in \End(\alg)$ such that
	\begin{equation}
		L \circ m_{\alg} \circ ( \phi_- \tp \id )
		= m_{\alg} \circ (\phi_- \tp L),
		\label{eq:counterterm-scalars}
	\end{equation}
	that is to say, $L$ is linear over the counterterms, we have
	\begin{equation}
		\bar{\phi} \circ B_+ = L \circ \phi_+.
		\label{eq:rbar-cocycle}
	\end{equation}
\end{satz}
\begin{proof}
	This is a straightforward consequence of the cocycle property of $B_+$:
	\begin{align*}
		\bar{\phi} \circ B_+
		&= \left( \phi_- \convolution \phi - \phi_- \right) \circ B_+
		 = m_{\alg} \circ (\phi_- \tp \phi) \circ \left[ (\id \tp B_+) \circ \Delta + B_+ \tp \1 \right] 
		 		- \phi_- \circ B_+ \\
		&= \phi_- \convolution \left( \phi \circ B_+ \right) 
		 = \phi_- \convolution \left( L \circ \phi \right)
		\urel{\eqref{eq:counterterm-scalars}} L \circ \left( \phi_- \convolution \phi \right)
		 = L \circ \phi_+ \qedhere.
	\end{align*}
\end{proof}
As the counterterms $\toyZ$ of our model are independent of $s$, they can be moved out of the integrals in \eqref{eq:toymodel} and \eqref{eq:counterterm-scalars} is fulfilled indeed. This is a general feature of {\qfts}: The counterterms to not depend on any external variables\footnote{Even if the divergence of a Feynman graph does depend on external momenta as happens for higher degrees of divergence, the Hopf algebra is defined such that the counterterms are evaluations on certain \emph{external structures}, given by distributions in \cite{CK:RH1}. So in any case, $\phi_-$ maps to scalars.}.

The significance of \eqref{eq:rbar-cocycle} lies in the expression of the renormalized $\toyR[0][](t)$ for a tree $t=B_+(\f)$ only in terms of the renormalized value $\toyR(\f)$. This allows for inductive proofs of properties of $\toyR$ and also $\toyR[][0]$, without having to consider the unrenormalized Feynman rules or their counterterms at all.

\subsection{Finiteness}
\label{sec:finiteness}

\begin{proposition}\label{satz:finiteness}
	The physical limit $\toyR[s][0]$ exists and maps $H_R$ into polynomials $\K[\ln \tfrac{s}{\rp}]$.
\end{proposition}

\begin{proof}
	We proceed inductively from $\toyR[s][0](\1) = 1$ and as $\toyR[][0]$ is a character only need to consider trees $t=B_+(\f)$ in the induction step. Hence for this $\f\in\forests$ we already know that $\toyR[\zeta][0](\f) \in \bigo{\ln^N \zeta}$ for some $N\in\N_0$ such that dominated convergence yields
	\settowidth{\wurelwidth}{\eqref{eq:rbar-cocycle}}
	\begin{align*}
		\toyR[s][0] (t)
		&\wurel{\eqref{eq:rbar-cocycle}} \lim_{\reg \rightarrow 0} (\id - \momsch{\rp}) \left[ s \mapsto \int_0^{\infty} \frac{f(\zeta/s)}{s} \zeta^{-\reg}\ \toyR[\zeta](\f)\ \dd\zeta \right] \nonumber\\
		&\wurel{} \lim_{\reg \rightarrow 0} \int_0^{\infty} \left[ \tfrac{f(\zeta/s)}{s} - \tfrac{f(\zeta/\rp)}{\rp} \right] {\zeta}^{-\reg}\ \toyR[\zeta] (\f)\ \dd\zeta
		\wurel{} \int_0^{\infty} \left[ \tfrac{f(\zeta/s)}{s} - \tfrac{f(\zeta/\rp)}{\rp} \right]\ \toyR[\zeta][0](\f)\ \dd\zeta,
	\end{align*}
	recalling the term in square brackets to be from $\bigo{\zeta^{-1-\varepsilon}}$ as in \eqref{eq:momsch-example}. This proves the cancellation of all $\reg$-poles in $\toyR[s](t)$ and we identify $\toyR[s][0](t)$ with the $\propto \reg^0$ term, which is a polynomial in $\ln s$ and $\ln\rp$ of degree $\abs{t}$ by inspection of \eqref{eq:toymodel-mellin}: Each such logarithm comes with a factor $\reg$ (expanding $s^{-\reg}$) which needs to cancel with a pole $\tfrac{c_{-1}}{\reg\abs{t_v}}$ from some $F(\reg\abs{t_v})$ in order to contribute to the $\propto \reg^0$ term.
	Finally the substitution $\zeta \mapsto \zeta\rp$ gives
	\begin{equation}
		\label{eq:BPHZ}
		\toyR[s][0](t)
		= \int_0^{\infty} \left[ \frac{f(\zeta\tfrac{\rp}{s})}{\tfrac{s}{\rp}} - f(\zeta) \right]\ \toyR[\rp\zeta][0](\f)\ \dd\zeta,
	\end{equation}
	hence by induction $\toyR[\zeta\rp][0]$ only depends on $\zeta$ and $\toyR[s][0]$ is a function of $\frac{s}{\rp}$ only.
\end{proof}
Using \eqref{eq:BPHZ}, the physical limit of the renormalized Feynman rules can be obtained inductively by convergent integrations after performing the subtraction at $s=\rp$ on the integrand, in particular without the need of any regulator. Therefore $\toyR[][0]$ is independent of the choice of regularization prescription, so employing a \emph{cutoff} regulator or \emph{dimensional regularization} yields the same renormalized result in the physical limit.

\section{The Hopf algebra of polynomials}
\label{sec:polynomials}

We summarize relevant properties of the polynomials, focusing on their Hochschild cohomology (the relevance of $\polyint$ was already mentioned in \cite{CK:NC}). First observe
\begin{lemma}
	Requiring $\Delta (x) = x\tp \1 + \1 \tp x$ induces a unique Hopf algebra structure on the polynomials ${\K}[x]$. It is graded by degree, connected, commutative and cocommutative with $\Delta \left( x^n \right) = \sum_{i=0}^n \binom{n}{i} x^i \tp x^{n-i}$ and the primitive elements are $\Prim \left( {\K}[x] \right) = \K \cdot x $.
\end{lemma}
The \emph{integration operator} $\polyint\!: x^n \mapsto \frac{1}{n+1}x^{n+1}$ furnishes a cocycle $\polyint \in \HZ[1]({\K}[x])$ as
	\begin{align*}
		\Delta \polyint \left( \frac{x^n}{n!} \right)
		&= \Delta \left( \frac{x^{n+1}}{(n+1)!} \right)
		= \sum_{k=0}^{n+1} \frac{x^k}{k!} \tp \frac{x^{n+1-k}}{(n+1-k)!} \\
		&= \frac{x^{n+1}}{(n+1)!} \tp \1 + \sum_{k=0}^n \frac{x^k}{k!} \tp \polyint \left(\frac{x^{n-k}}{(n-k)!} \right)
		= \left[ \polyint \tp \1 + \left(\id \tp \polyint \right) \circ \Delta \right] \left( \frac{x^n}{n!} \right),
	\end{align*}
and is not a coboundary since $\polyint 1 = x \neq 0$. In fact it generates the cohomology by
\begin{satz}
	$\HH[1] (\K[x]) = \K \cdot [ \polyint ]$ is one-dimensional as the 1-cocycles of $\K[x]$ are
	\begin{equation}
		\HZ[1](\K[x]) 
		= \K \cdot \polyint 
			\ \oplus\ 
			\dH \left( \K[x]' \right)
		= \K \cdot \polyint 
			\ \oplus\ 
			\HB[1](\K[x]).
		\label{eq-polys-cycles}
	\end{equation}
\end{satz}

\begin{proof}
	For an arbitrary cocycle $L \in \HZ[1] (\K[x])$, lemma \ref{satz:cocycle-props} ensures $L(1) = x a_{-1}$ where $a_{-1} \defas \partial_0 L(1)$. Hence $\tilde{L} \defas L - a_{-1}\polyint \in \HZ[1]$ fulfils $\tilde{L}(1)=0$, so $L_0 \defas \tilde{L} \circ \int_0 \in \HZ[1]$ by
	\begin{align*}
		\Delta \circ L_0
		= (\id \tp \tilde{L}) \circ \Delta \circ \polyint + (\tilde{L} \tp 1) \circ \polyint
		= (\id \tp L_0) \circ \Delta + L_0 \tp 1 + \tilde{L}(1) \cdot \polyint.
	\end{align*}
	Repeating the argument inductively yields $a_n\defas\partial_0 L_n(1) = \partial_0 \circ L \circ \polyint^{n+1}(1)\in\K$ and $L_{n+1} \defas (L_n - a_n\polyint) \circ \polyint \in \HZ[1]$, so for any $n\in\N_0$ we may read off from
	\begin{align*}
		L\circ\polyint^n (1) 
		= a_{-1} \polyint^{n+1} (1) + \ldots + a_{n-2} \polyint^2 (1) + L_{n-1}(1)
		= a_{-1} \polyint \left( \polyint^n 1 \right) + \sum_{j=0}^{n-1} a_j \polyint^{n-j} (1)
	\end{align*}
	that indeed $L=a_{-1}\polyint + \dH\alpha$ for the functional $\alpha\defas\partial_0\circ L \circ \polyint$ with $\alpha (\frac{x^n}{n!}) = a_j$.
\end{proof}
\begin{lemma}\label{satz:poly-coboundaries}
	Up to subtraction
	$
		P
		= \dH\counit
		= \id - \ev_0\!:
		{\K}[\x] \twoheadrightarrow \ker\counit
		= \x{\K}[\x]
	$
	of the constant part, direct computation exhibits $\dH \alpha$ for any $\alpha \in \K[\x]'$ as the differential operator
	\begin{equation}\label{eq:poly-coboundaries}
		\dH \alpha
		= P \circ \sum_{n\in\N_0} \alpha \left( \tfrac{x^n}{n!} \right) \partial^n
		\in \End({\K}[\x]).
	\end{equation}
\end{lemma}

\begin{lemma}\label{satz:poly-characters}
	As any character $\phi \in \chars{{\K}[x]}{\K}$ of $\K[x]$ is fixed by $\lambda\defas \phi(x)$, they are the group $\chars{{\K}[x]}{\K} = \setexp{\ev_{\lambda}}{\lambda \in \K}$ of evaluations (the counit $\counit = \ev_0$ equals the neutral element)
	\begin{equation}\label{eq:poly-characters}
		\K[x] \ni p(x)
		\mapsto \ev_{\lambda} (p) \defas p(\lambda)
		\quad \text{with the product} \quad
		\ev_a \convolution \ev_b = \ev_{a+b}.
	\end{equation}
\end{lemma}
\begin{proof}
	Note 
	$
		\left[ \ev_a \convolution \ev_b \right] \left( x^n \right)
		= { \left[ \ev_a(1) \cdot \ev_b(x) + \ev_a(x) \cdot \ev_b(1) \right] }^n
		= { ( b + a ) }^n
	$.
\end{proof}
\begin{lemma}\label{satz:poly-log}
	The isomorphism $(\K,+)\ni a \mapsto \ev_a \in \chars{\K[x]}{\K}$ of groups is generated by the functional $\partial_0 = \ev_0\circ\partial\in\infchars{\K[x]}{\K}$, meaning $\log_{\convolution} \ev_a = a \partial_0$ and $\ev_a = \exp_{\convolution} (a\partial_0)$.
\end{lemma}
\begin{proof}
	Expanding the exponential series reveals $\exp_{\convolution}(a\partial_0) (x^n) = a^n$ as a direct consequence of $\partial_0^{\convolution k} = \counit \circ \partial^{\convolution k} = \counit \circ \partial^k = \partial_0^k$, while we appreciate $\partial^{\convolution k} = \partial^k$ inductively:
\begin{equation*}
	\partial \convolution \partial^k \left(\tfrac{x^n}{n!}\right)
	= \sum_{j=0}^n \left(\partial \tfrac{x^{n-j}}{(n-j)!}\right) \left(\partial^k \tfrac{x^j}{j!} \right)
	= \sum_{j=k}^{n-1} \frac{x^{n-1-k}}{(n-j-1)! (j-k)!}
	= \frac{x^{n-k-1}}{(n-k-1)!} 
	.\qedhere
\end{equation*}
\end{proof}

\subsection{Feynman rules induced by cocycles}
\label{sec:toymodel-rules-by-cocycles}

Let $\toyphy\!: H_R\rightarrow\K[\x]$ denote the polynomials that evaluate to the renormalized Feynman rules $\toyR[s][0] = \ev_{\scalelog} \circ \toyphy$ at $\scalelog=\ln\frac{s}{\rp}$. We state
\begin{satz}\label{satz:toymodel-universal}
	The renormalized Feynman rules $\toyphy = \unimor{\toycc}$ arise out of the universal property of theorem \ref{satz:H_R-universal}, where the coefficients $\coeff{n}$ of \eqref{eq:mellin-trafo} determine the cocycle
	\begin{equation}\label{eq:cocycle-toymodel}
		\toycc \defas -\coeff{-1} \polyint + \dH \toyform \in \HZ[1](\K[\x])
		\quad\text{with}\quad
		\toyform \left( \x^n \right) \defas n!\, (-1)^n \coeff{n}
		\quad\text{for any $n\in \N_0$.}
	\end{equation}
\end{satz}

\begin{proof}
	We may set $\rp=1$ and produce logarithms of subdivergences by differentiation, exploiting analyticity of $\reg F(\reg)$ and $\frac{s^{-\reg}-1}{\reg}$ at $\reg=0$ we obtain
	\begin{align*}
		&\lim_{\reg \rightarrow 0} (\id - \momsch{1}) \left[ s \mapsto \int_0^{\infty} f(\zeta) {\left( s\zeta \right)}^{-\reg} \ln^n \left( s\zeta \right)\;\dd\zeta \right]
		= {\left( -\frac{\partial}{\partial \reg} \right) }_{\reg=0}^n
	%	\restrict{\left( -\partial_z \right)^n}{0}
		(\id - \momsch{1}) \int_0^{\infty} f(\zeta) {\left( s\zeta \right)}^{-\reg}\; \dd\zeta \displaybreak[0]\\
		&%= {\left( -\frac{\partial}{\partial \reg} \right)}_{\reg=0}^n \left( s^{-\reg} - 1 \right) \int_0^{\infty} f(\zeta) \zeta^{-\reg}\; \dd\zeta
		 = {\left( -\frac{\partial}{\partial \reg} \right)}_{\reg=0}^n 
		 \left[
		 		\frac{s^{-\reg}-1}{\reg}
				\cdot
				\reg F(\reg)
		 \right]
\hide{
		&= {(-1)}^n \sum_{k=0}^n \binom{n}{k} 
		 		\left\{ {\left( \frac{\partial}{\partial \reg} \right)}_{\reg=0}^k \frac{s^{-\reg}-1}{\reg} \right\} 
				\cdot 
				\left\{ {\left( \frac{\partial}{\partial \reg} \right)}_{\reg=0}^{n-k} \big[ \reg F(\reg) \big] \right\} \\
}
		= {(-1)}^n \sum_{k=0}^n \binom{n}{k} k! \frac{{\left( - \ln s \right)}^{k+1}}{(k+1)!} (n-k)!\, \coeff{n-k-1} \displaybreak[0]\\
		%= \ev_{\ln s} \left[ \sum_{k=0}^n \frac{n!\, \x^{k+1}}{(k+1)!} {(-1)}^{n-k-1} \coeff{n-k-1}  \right] \\
%		&= - \frac{c_{-1}}{n+1} \ln^{n+1} s + \sum_{k=0}^{n-1} \frac{n!}{(k+1)!} (-1)^{n-k-1} c_{n-k-1} \ln^{k+1} s \\
%		&= - c_{-1} \int_0 \ln^n s + \sum_{l=1}^n \frac{n!}{l!} (-1)^{n-l} c_{n-l} \ln^l s
		&= \ev_{\ln s} \left[-\coeff{-1} \frac{\x^{n+1}}{n+1} + \sum_{i=1}^n \binom{n}{i} \x^i {(-1)}^{n-i} \coeff{n-i} (n-i)! \right]
%		= -c_{-1} \int_0 \ln^n s + \partial \toyform \left( \ln^n s \right).
		 = \ev_{\ln s} \circ \toycc \left( \x^n \right).\tag{$\ast$}
	\end{align*}
	By linearity we can replace $\ln^n (s\zeta)$ in the integrand by any polynomial to prove \ref{satz:toymodel-universal} inductively: As $\toyphy$ and $\unimor{\toycc}$ are algebra morphisms, it suffices to consider a tree $t=B_+(\f)$ for a forest $\f\in\forests$ already fulfilling $\toyphy(\f) = \unimor{\toycc}(\f)$ in the induction step
	\settowidth{\wurelwidth}{\eqref{eq:BPHZ}}
	\begin{align*}
		\toyR[s][0] (t) &
		\wurel{\eqref{eq:BPHZ}} \lim_{\reg \rightarrow 0} (\id - \momsch{1}) \left[
				s \mapsto \int_0^{\infty} f(\zeta) {\left( s\zeta \right)}^{-\reg}\ \ev_{\ln s\zeta} \circ \toyphy(\f)\ \dd\zeta
			\right] \\
		&\wurel{$(\ast)$} \ev_{\scalelog} \circ \toycc \left[ \toyphy (\f) \right]
		= \ev_{\scalelog} \circ \toycc \circ \unimor{\toycc} (\f)
		\urel{\ref{satz:H_R-universal}} \ev_{\scalelog} \circ \unimor{\toycc} \circ B_+ (\f)
		= \ev_{\scalelog} \circ \unimor{\toycc} (t),
	\end{align*}
	where the convergence of \eqref{eq:BPHZ} allows to reintroduce $\zeta^{-\reg}$ into the integrand.
\end{proof}

\begin{korollar}\label{satz:toyphy-hopfmor}
	As $\toycc$ is a cocycle, by theorem \ref{satz:H_R-universal} the physical limit $\toyphy\!: H_R \rightarrow \K[\x]$ of the renormalized Feynman rules \eqref{eq:toymodel} is a morphism of Hopf algebras.
\end{korollar}
This key property will imply the renormalization group in the sequel. For now observe the simple and explicit combinatorial recursion \ref{satz:toymodel-universal}, expressing $\toyphy$ in terms of the Mellin transform coefficients without any need for series expansions in $\reg$, as shown in
\begin{beispiel}\label{ex:toymodel-universal}
	Using \eqref{eq:cocycle-toymodel} we rederive
$
	\toyphy \left( \tree{+-} \right)
	 = \unimor{\toycc} \circ B_+ (\1)
	 = L(1)
	 = - \coeff{-1}\, \x
$ and also
\begin{align*}
	\toyphy \left( \tree{++--} \right)
	&= \unimor{\toycc} \circ B_+ \left( \tree{+-} \right)
	 = \toycc \circ \unimor{\toycc} \left( \tree{+-} \right)
	 = \left[ -\coeff{-1}\polyint + \dH\toyform \right] \left( -\coeff{-1} \x \right)
	 = \coeff[2]{-1} \frac{\x^2}{2} - \coeff{-1} \coeff{0}\, \x, \\
\hide{
	 \toyphy \left( \tree{+++---} \right)
	&= \left[ -\coeff{-1}\polyint + \dH\toyform \right] \left( \coeff[2]{-1} \frac{\x^2}{2} - \coeff{-1} \coeff{0}\, \x \right)
	= - \coeff[3]{-1} \frac{\x^3}{6} + \coeff[2]{-1} \coeff{0} \frac{\x^2}{2}+ \coeff[2]{-1} \toyform(\x)\,\x \\
	 &\quad + \toyform(1) \left( \coeff[2]{-1} \frac{\x^2}{2} - \coeff{-1} \coeff{0}\, \x \right) 
	 = - \coeff[3]{-1} \frac{\x^3}{6} + \coeff[2]{-1} \coeff{0}\, \x^2 - \left( \coeff{-1} \coeff[2]{0} + \coeff[2]{-1} \coeff{1} \right) \x \\
}%
	\toyphy \left( \tree{++-+--} \right)
	&= \unimor{\toycc} \circ B_+ \left( \tree{+-}\tree{+-} \right)
	 = \toycc \circ \unimor{\toycc} \left( \tree{+-}\tree{+-} \right) 
	 = \left[ -\coeff{-1}\polyint + \dH\toyform \right] \left\{ {\left( -\coeff{-1}\,\x \right) }^2 \right\} \\
	&= -\coeff[3]{-1} \frac{\x^3}{3} + \coeff[2]{-1} \left[ \toyform(1)\,\x^2 + 2\toyform(\x)\,\x \right]
	 = -\coeff[3]{-1} \frac{\x^3}{3} + \coeff[2]{-1} \coeff{0}\,\x^2 - 2\coeff[2]{-1}\coeff{1}\,\x.
\end{align*}
\end{beispiel}
Defining $\tilde{F}(\reg)\defas F(\reg)-\frac{\coeff{-1}}{\reg} = \sum_{n\in\N_0} \coeff{n}\reg^n$, \eqref{eq:poly-coboundaries} uncovers $\dH\toyform = P \circ \tilde{F}(-\partial_{\x})$ and under the convention $\partial_{\x}^{-1}\defas\polyint$	we may thus write $\toycc = P\circ F(-\partial_{\x})$.

\begin{korollar}\label{satz:toymodel-leading-log}
	As in $\toyform$ only $-\coeff{-1}\polyint$ increases the degree in $\x$, the highest order (called \emph{leading $\log$}) of $\toyphy$ is the tree factorial (note the analogy to \eqref{eq:toymodel-leading-pole}): For any forest $\f \in \forests$,
	\begin{equation}
		\toyphy (\f) \in
		\unimor{\left[-\coeff{-1}\polyint\right]} (\f) + \bigo{\x^{\abs{\f}-1}}
		\urel{\ref{ex:tree-factorial}}
		\frac{ {\left( - \coeff{-1} \x \right)}^{\abs{\f}}}{\f!} + \K[\x]_{<\abs{\f}}.
	\end{equation}
\end{korollar}

\subsection{Feynman rules as Hopf algebra morphisms}

As $\toyphy\!: H_R\rightarrow {\K}[\x]$ is a morphism of Hopf algebras, the induced map $\chars{{\K}[\x]}{\K}\rightarrow \chars{H_R}{\K}$ given by $\ev_a \mapsto \toyphy[a]\defas \ev_a \circ \toyphy$ becomes a morphism of groups. In particular note
\begin{korollar}\label{satz:rge}
	Using \eqref{eq:poly-characters} we obtain the \emph{renormalization group equation} (as in \cite{Kreimer:ChenII})
	\begin{equation}
		\toyphy[a]	\convolution	\toyphy[b]
		= \toyphy[a+b],
	\quad\text{for any}\quad
	a,b \in \K.
\end{equation}
\end{korollar}
Before we obtain the generator of this one-parameter group in corollary \ref{korollar:toylog}, note how this result gives non-trivial relations between individual trees (graphs) like
\begin{align*}
	\toyphy[a] \convolution \toyphy[b] \left( \tree{++--} \right)
	&= \toyphy[a] \left( \tree{++--} \right)
		+ \toyphy[a] \left( \tree{+-} \right) \toyphy[b] \left( \tree{+-} \right)
		+ \toyphy[b] \left( \tree{++--} \right) \\
	&\urel{\eqref{toyR-physical:(())}}
		\coeff[2]{-1} \frac{a^2+b^2}{2}
		- \coeff{-1}\coeff{0} \left(a + b \right)
		+ \coeff[2]{-1}ab
	\urel{\eqref{toyR-physical:(())}}
		\toyphy[a+b] \left( \tree{++--} \right).
\end{align*}
\begin{proposition}
	Let $H$ be any connected bialgebra and $\phi\!: H \rightarrow \K[\x]$ a morphism of bialgebras.\footnote{This already implies $\phi$ to be a morphism of Hopf algebras.} Then $\log_\convolution \phi$ is given by the linear term in $\x$ through
	\begin{equation}
		\log_{\convolution} \phi
		= \x \cdot \partial_0 \circ \phi.
	\end{equation}
\end{proposition}
\begin{proof}
	Letting $\phi\!: C\rightarrow H$ and $\psi\!: H\rightarrow \alg$ denote morphisms of coalgebras and algebras, exploiting
$
	\left( \psi \circ \phi - \unit_{\alg} \circ \counit_C \right)^{\convolution n}
		= \psi \circ \left( \phi - \unit_H\circ\counit_H \right)^{\convolution n}
		= \left( \psi - \unit_{\alg}\circ\counit_H \right)^{\convolution n} \circ \phi 
$
in \eqref{eq:log-exp} proves $(\log_{\convolution} \psi) \circ \phi = \log_{\convolution} (\psi \circ \phi) = \psi \circ \log_{\convolution} \phi$. Now set $\psi=\ev_a$ and use lemma \ref{satz:poly-log}.
\end{proof}
\begin{beispiel}\label{ex:intrules-exp}
	In the leading-$\log$ case \eqref{eq:int-rules} we read off $\partial_0 \circ \intrules = Z_{\tree{+-}} \in \infchars{H_R}{\K}$ where $Z_{\tree{+-}} (\f) \defas \delta_{\f,\tree{+-}}$. Comparing $\intrules = \exp_{\convolution}(xZ_{\tree{+-}})$ with \eqref{eq:int-rules} shows $\abs{\f}! = \f! \cdot Z_{\tree{+-}}^{\convolution \abs{\f}}(\f)$, hence\footnote{This combinatorial relation among tree factorials, noted in \cite{Kreimer:ChenII}, thus drops out of $\Delta\intrules = (\intrules \tp \intrules) \circ \Delta$.}
	\begin{equation*}
		\frac{\abs{\f}}{\f!}
		=	\frac{1}{\left( \abs{\f} -1 \right)!} \sum_{\f} Z_{\tree{+-}}(\f_1)Z_{\tree{+-}}^{\convolution \abs{\f}-1}(\f_2)
		= \sum_{\f:\ \f_1=\tree{+-}}
			\frac{1}{\abs{\f_2}!}Z_{\tree{+-}}^{\convolution \abs{\f_2}} (\f_2)
		= \sum_{\f:\ \f_1=\tree{+-}}
			\frac{1}{\f_2!}.
	\end{equation*}
\end{beispiel}
\begin{korollar}\label{korollar:toylog}
	The character $\toyphy$ is fully determined by the \emph{anomalous dimension}
	\begin{equation}\label{eq:toylog}
		H_R' \supset
		\infchars{H_R}{\K} \ni \toylog
		\defas -\partial_0\: \circ\: \toyphy
		\quad\text{such that}\quad
		\toyphy
			= \exp_{\convolution}(-\x\cdot\toylog)
			= \sum_{n\in\N_0} \frac{\toylog^{\convolution n}}{n!} (-\x)^n.
	\end{equation}
\end{korollar}
An analogous phenomenon happens with the counterterms in the minimal subtraction scheme: The first order poles $\propto \reg^{-1}$ alone already determine the full counterterm via the \emph{scattering formula} proved in \cite{CK:RH2}. However, \eqref{eq:toylog} is much simpler as illustrated in
\begin{beispiel}
	Reading off
$
	\toylog \left( \tree{+-} \right)
	= \coeff{-1}
$,
$
	\toylog \left( \tree{++--} \right)
	= \coeff{-1}\coeff{0}
$
%$
%	\toylog \left( \tree{+++---} \right)
%	= \coeff{-1}\coeff[2]{0} + \coeff[2]{-1} \coeff{1}
%$
and
$
	\toylog \left( \tree{++-+--} \right)
	= 2\coeff[2]{-1}\coeff{1}
$
from the example \ref{ex:toymodel-universal} above, corollary \ref{korollar:toylog} determines the higher powers of $\x$ through
\begin{align*}
	\toyphy \left( \tree{++--} \right)
	&%= \exp_{\convolution} (-\x\cdot\toylog) \left( \tree{++--} \right)
	 \urel{\eqref{eq:log-exp}} \left[ e - \x\toylog + \x^2\frac{\toylog\convolution\toylog}{2} \right] \left( \tree{++--} \right)
	 = 0 - \x\toylog\left( \tree{++--} \right) + \x^2 \frac{\toylog^2\left( \tree{+-} \right)}{2} 
	 = - \coeff{-1}\coeff{0}\,\x + \coeff[2]{-1} \frac{\x^2}{2}, \\
	\toyphy \left( \tree{++-+--} \right)
	%&= \exp_{\convolution} (\x\cdot\toylog) \left( \tree{++-+--} \right)
	% \urel{\eqref{eq:exp-conv}} \left[ e + \x\toylog + \x^2\frac{\toylog \convolution \toylog}{2} + \x^3\frac{\toylog \convolution \toylog \convolution \toylog}{6} \right] \left( \tree{++-+--} \right) \\
	&= 0 - \x \toylog \left( \tree{++-+--} \right) 
		+ \x^2\frac{\toylog \tp \toylog}{2} \left( 2 \tree{+-} \tp \tree{++--} + \tree{+-}\tree{+-} \tp \tree{+-} \right)
		- \x^3\frac{\toylog \tp \toylog \tp \toylog}{6} \left( 2\tree{+-} \tp \tree{+-} \tp \tree{+-} \right) \\
	&= -\toylog^3\left(\tree{+-}\right) \frac{\x^3}{3} + \x^2 \toylog\left(\tree{+-}\right) \toylog\left(\tree{++--}\right)  - 2\coeff[2]{-1}\coeff{1}\,\x
	 = -\coeff[3]{-1} \frac{\x^3}{3} + \coeff[2]{-1} \coeff{0}\,\x^2 - 2\coeff[2]{-1}\coeff{1}\,\x.
\end{align*}
\end{beispiel}
Note how the fragment $\tree{+-}\tree{+-} \tp \tree{+-}$ of $\Delta \left( \tree{++-+--} \right)$ does not contribute to the quadratic terms $\frac{\x^2}{2} \toylog \convolution \toylog$, as $\toylog$ vanishes on products. We will exploit this in \eqref{eq:perturbation-convolution-inf} of section \ref{sec:propagator-coupling} and close with a method of calculating $\toylog$ emerging from
\begin{lemma}
	From
	$
		\toylog\circ B_+ 
		= -\partial_0 \circ L \circ \toyphy 
%		= \ev_0 \circ (\coeff{-1} + \restrict{\reg\tilde{F}(\reg)}{-\partial_{\x}}) \circ\exp_{\convolution}(-\x\toylog)
		= \ev_0 \circ [\reg F(\reg)]_{-\partial_{\x}} \circ\exp_{\convolution}(-\x\toylog)
	$
	we obtain the inductive formula
	$
		\toylog \circ B_+
		= \sum_{n\in\N_0} \coeff{n-1} \toylog^{\convolution n}.
	$
\end{lemma}
\begin{beispiel}
	We can recursively calculate 
	$
		\toylog\left(\tree{+-} \right)
		= \coeff{-1} \counit(\1) = \coeff{-1}
	$,
	similarly also
	\begin{align*}
		\toylog\left( \tree{++--} \right)
		&= \coeff{-1}\counit\left( \tree{+-} \right) + \coeff{0} \toylog\left( \tree{+-} \right)
		= \coeff{-1}\coeff{0}, \\
		\toylog\left( \tree{+++---} \right)
		&= \coeff{-1}\counit\left( \tree{++--} \right) + \coeff{0}\toylog\left( \tree{++--} \right) + \coeff{1}\toylog\convolution\toylog\left( \tree{++--} \right)
		= \coeff{-1}\coeff[2]{0} + \coeff{1} \left[ \toylog\left( \tree{+-} \right) \right]^2
		= \coeff{-1}\coeff[2]{0} + \coeff[2]{-1}\coeff{1}, \\
		\toylog\left( \tree{++-+--} \right)
		&= \coeff{-1}\counit\left( \tree{+-} \tree{+-} \right)
			+\coeff{0}\toylog\left( \tree{+-} \tree{+-} \right)
			+\coeff{1}\toylog\convolution\toylog\left( \tree{+-} \tree{+-} \right)
		= 2\coeff{1} \left[ \toylog\left( \tree{+-} \right) \right]^2
		= 2\coeff[2]{-1}\coeff{1}
		\quad\text{and so on.}
	\end{align*}
\end{beispiel}

\section{Dyson-Schwinger equations and correlation functions}
\label{sec:DSE}

We now study the implications for the \emph{correlation functions} \eqref{eq:correlation} as formal power series in the \emph{coupling constant} $\coupling$. For simplicity we restrict to a single equation and refer to \cite{Yeats} for systems. With detailed treatments in \cite{BergbauerKreimer,Foissy:DSE}, for our purposes suffices
\begin{definition}
	To a parameter $\powdep\in\K$ and a family of cocycles $B_{\cdot}\!: \N\rightarrow \HZ[1](H_R)$ we associate the \emph{combinatorial Dyson-Schwinger equation}\footnote{
	As $x_0=\1$, for arbitrary $p$ the series $\left[ X(\coupling) \right]^{p} \defas \sum_{n\in\N_0} \binom{p}{n} \left[ X(\coupling)-\1 \right]^n \in H_R[[\coupling]]$ is well defined.}
	\begin{equation}
		X(\coupling)
		= \1 + \sum_{n\in\N} \coupling^n B_n \left( X^{1+n\powdep}(\coupling) \right).
		\label{eq:dse}
	\end{equation}
\end{definition}
\begin{lemma}\label{satz:dse-coprod}
	As \emph{perturbation series}
	$
		X(\coupling) 
		= \sum_{n\in\N_0} 
				x_n
				\coupling^n
		\in H_R[[\coupling]]
	$,
	equation \eqref{eq:dse} has a unique solution. It begins with $x_0=\1$ while $x_{n+1}$ is determined recursively from $x_0, \ldots, x_n$. These coefficients generate a Hopf sub algebra, explicitly we find\footnote{
		A proof of \eqref{eq:dse-coprod} may be found in \cite{Foissy:DSE} and \cite{Foissy:Systems,Foissy:GeneralSystems} study systems of Dyson-Schwinger equations.
	}
	\begin{equation}\label{eq:dse-coprod}
		\Delta X(\coupling)
		= \sum_{n\in\N_0} \left[ X(\coupling) \right]^{1+n\powdep} \tp \coupling^n x_n
		\in (H_R\tp H_R) [[\coupling]].
	\end{equation}
\end{lemma}

\begin{beispiel}\label{ex:DSE-propagator}
	In \cite{Kreimer:ExactDSE,Panzer:Master}, $X(\coupling) = \1 - \coupling B_+ \left( \frac{1}{X(\coupling)} \right)$ features $\powdep=-2$, summing all trees
\begin{equation*}\begin{split}
	X(\coupling)\in
	&\	\1
		- \tree{+-}\ \coupling
		- \tree{++--}\ \coupling^2
		- \left(
				\tree{+++---} + \tree{++-+--}
			\right) \coupling^3
		- \left(
					\tree{++++----} 
				+	\tree{+++-+---}
				+2\tree{++-++---}
				+	\tree{++-+-+--} 
			\right) \coupling^4 \\
	&	- \left(
					\tree{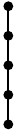} + \tree{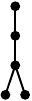} 
				+2\tree{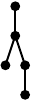} + \tree{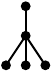}
				+ \tree{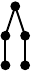} + 2\tree{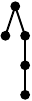} 
				+2\tree{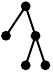} + 3\tree{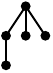} 
				+ \tree{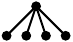}
			\right) \coupling^5
		+ \coupling^6 H_R[[\coupling]]
\end{split}\end{equation*}
with a symmetry factor. Physically these correspond to (Yukawa) propagators 
\begin{multline*}
	 \1
		-\Graph{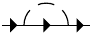} \coupling
		-\Graph{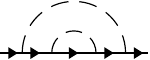} \coupling^2
		-\left(
				\Graph{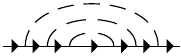} + \Graph{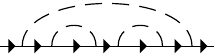}
			\right) \coupling^3 \\
		-\left(
				\Graph[0.7]{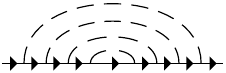} + \Graph[0.7]{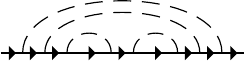}
				+\Graph[0.7]{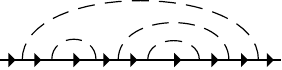} + \Graph[0.7]{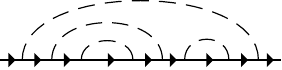} + \Graph[0.7]{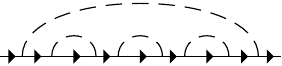}
			\right) \coupling^4
		+\bigo{\coupling^5},
\end{multline*}
arising from insertions of the one-loop graph $\Graph{+-}$ into itself.
\end{beispiel}

\begin{definition}\label{satz:correlation}
	The \emph{correlation function} $G(\coupling)$ evaluates the renormalized Feynman rules $\toyphy\!: H_R\rightarrow\K[\scalelog]$ on the perturbation series $X(\coupling)$, yielding the formal power series
	\begin{equation}\label{eq:correlation}
		G (\coupling)
		\defas \toyphy \circ X(\coupling)
		= \sum_{n\in\N_0} \toyphy (x_n) \coupling^n
		\in \left( \K[\scalelog] \right) [[\coupling]].
	\end{equation}
	We call 
	$
		\widetilde{\toylog}(\coupling)
		\defas \toylog \circ X(\coupling)
		= \restrict{-\partial_{\scalelog}}{0} G(\coupling)
		\in {\K}[[\coupling]]
	$ the \emph{physical anomalous dimension}.
\end{definition}
\begin{beispiel}\label{ex:intrules-correlation}
	The Feynman rules $\intrules$ from \eqref{eq:int-rules} result in the convergent series $G(\coupling) = \sqrt{1-2\coupling\scalelog}$ and 
	$
		\widetilde{\toylog}(\coupling) 
		\urel{\ref{ex:intrules-exp}} - Z_{\tree{+-}} \circ X(\coupling)
		= \coupling
	$
	for the propagator of example \ref{ex:DSE-propagator}. Perturbatively,
\begin{align*}
	G(\coupling)
	&= 1
		-\frac{(\coupling \scalelog)}{\tree{+-}\,!}
		-\frac{(\coupling \scalelog)^2}{\tree{++--}\,!}
		-\frac{(\coupling \scalelog)^3}{\tree{+++---}\,!} 
		-\frac{(\coupling \scalelog)^3}{\tree{++-+--}\,!}
		- \ldots
	 = 1
		-\coupling\scalelog
		-\frac{1}{2} (\coupling\scalelog)^2
		-\frac{1}{2} (\coupling\scalelog)^3 
%		-\frac{5}{8} (\coupling\scalelog)^4 
%		-\frac{7}{8} (\coupling\scalelog)^5
		+\bigo{{(\coupling\scalelog)}^4}.
\end{align*}
\end{beispiel}

\subsection{Propagator coupling duality}
\label{sec:propagator-coupling}

The Hopf subalgebra of the perturbation series allows to calculate convolutions in
\begin{lemma}\label{lemma:perturbation-convolutions}
	Let $\psi\in\infchars{H_R}{\alg}$ denote an infinitesimal character, $\Psi\in\chars{H_R}{\alg}$ a character and $\lambda\in\Hom(H_R,\alg)$ a linear map. Then (in suggestive notation)
	\begin{align}
		(\Psi \convolution \lambda) \circ X(\coupling)
		=& \left[ \Psi \circ X(\coupling)\right] 
				\cdot \lambda\circ X\left(
					\coupling 
					\left[ \Psi \circ X(\coupling) \right]^{\powdep}
				\right)
		\label{eq:perturbation-convolution-char}\\
		\defas&
				\left[ \Psi \circ X(\coupling)\right] 
				\cdot \sum_{n\in\N_0} \lambda(x_n) \cdot \left(
					\coupling 
					\left[ \Psi \circ X(\coupling) \right]^{\powdep}
				\right)^n
		\in{\alg}[[\coupling]] \nonumber\\
		(\psi \convolution \lambda) \circ X(\coupling)
		=&	\left[ \psi\circ X(\coupling) \right] \cdot
				\left( \id + \powdep \coupling \partial_{\coupling} \right)
				\left[ \lambda \circ X(\coupling) \right]
		\in{\alg} [[\coupling]].
		\label{eq:perturbation-convolution-inf}
	\end{align}
\end{lemma}
\begin{proof}
	These are immediate consequences of lemma \ref{satz:dse-coprod}, for \eqref{eq:perturbation-convolution-inf} consider
	\begin{equation*}
				\psi \left( \left[ X(\coupling) \right]^{1+n\powdep} \right) \cdot
				%\lambda \left(x_n
					\coupling^n
				%\right) 
		=		\sum_{i\in\N_0}
				\binom{1+n\powdep}{i}
				\psi \left( \left[ X(\coupling) - \1 \right]^{i} \right)
				%\lambda \left( x_n 
					\coupling^n
				%\right)
		 =	\psi\left( X(\coupling) - \1 \right)
				\cdot
%				\sum_{n\in\N_0}
				(1+n\powdep) % \lambda\left( x_n 
					\coupling^n 
				%\right)
		. \qedhere
	\end{equation*}
\end{proof}

\begin{beispiel}
	Continuing \ref{ex:DSE-propagator} we deduce $Z_{\tree{+-}}^{\convolution 2} \left(X(\coupling)\right) = -\coupling(1-2\coupling\partial_{\coupling})(-\coupling) = -\coupling^2$ and
\begin{equation*}
	Z_{\tree{+-}}^{\convolution n+1} \left( X(\coupling) \right)
	\urel{\eqref{eq:perturbation-convolution-inf}}
	-\coupling^{n+1} (2n-1)(2n-3)\cdots(1)
	= -\coupling^{n+1} \frac{(2n)!}{2^n n!},
\end{equation*}
proving $\intrules(x_{n+1}) = -2^{-n} C_n \scalelog^{n+1}$ with the \emph{Catalan numbers} $C_n$ already noted in \cite{Kreimer:NonlinearDSE}. From \ref{ex:intrules-correlation} we find their generating function 
	$
		2\coupling \sum_{n\in\N_0} \coupling^n C_n 
		= 1-\sqrt{1-4\coupling}
	$.
\end{beispiel}

\begin{korollar}
	As $\toyphy$ is a morphism of Hopf algebras, for any $a,b\in\K$ we can factor
	\begin{equation}
		G_{a+b}(\coupling)
		= (\toyphy[a]	\convolution	\toyphy[b]) \circ X(\coupling)
		\urel{\eqref{eq:perturbation-convolution-char}}
			G_a (\coupling) \cdot G_b\left[ \coupling G_a^{\powdep}(\coupling) \right]
		= G_b (\coupling) \cdot G_a\left[ \coupling G_b^{\powdep}(\coupling) \right].
		\label{eq:propagator-coupling}
	\end{equation}
\end{korollar}
These functional equations of formal power series make sense for the non-perturbative correlation functions as well. Relating the scale- with the coupling-dependence, this integrated form of the renormalization group equation becomes infinitesimally
\begin{korollar}
	From $-\frac{\dd}{\dd\x}\ \toyphy = \toylog \convolution \toyphy = \toyphy \convolution \toylog$ or differentiating \eqref{eq:propagator-coupling} by $b$ at zero note
	\begin{equation}\label{eq:propagator-coupling-diff}
		G_{\scalelog}(\coupling) \cdot \widetilde{\toylog}\left[ \coupling G_{\scalelog}^{\powdep}(\coupling) \right]
		\urel{\eqref{eq:perturbation-convolution-char}}
		-\partial_{\scalelog} G_{\scalelog}(\coupling)
		\urel{\eqref{eq:perturbation-convolution-inf}}
		\widetilde{\toylog}(\coupling) \cdot \left( 1 + \powdep \coupling \partial_{\coupling} \right) G_{\scalelog}(\coupling).
	\end{equation}
\end{korollar}
The first of these equations generalizes the \emph{propagator coupling duality} in \cite{Kreimer:ExactDSE,Kreimer:NonlinearDSE}. For any fixed coupling $\coupling$, it expresses the correlation function as the solution of the ode
\begin{equation}
	-\frac{\dd}{\dd \scalelog} \ln G_{\scalelog} (\coupling)
	= \widetilde{\toylog} \left[ \coupling e^{\powdep \ln G_{\scalelog}(\coupling)} \right]
	\quad\text{with}\quad
	\ln G_0 (\coupling) = 0,
	\label{eq:propagator-coupling-ode}
\end{equation}
determining $G_{\scalelog}(\coupling)$ completely from $\widetilde{\toylog}(\coupling)$ in a non-perturbative manner as in \eqref{eq:rge-solved}.
\hide{
We may rewrite \eqref{eq:propagator-coupling-ode} as
\begin{equation*}
	\scalelog
	= \int_1^{G_{\scalelog}(\coupling)}
		\frac{\dd y}{y \tilde{\toylog}\left( \coupling\cdot y^{\powdep} \right)}.
\end{equation*}
\paragraph{Example}
	Consider the leading-$\log$-expansion of the propagator:
\begin{align*}
	G_{b}(\coupling) \cdot G_a\left[ \coupling G_b^{-2}(\coupling) \right]
	&=	\sqrt{1+2\coeff{-1}b\coupling} \cdot \sqrt{1+2\coeff{-1}a\coupling \left[ \sqrt{1+2\coeff{-1}b\coupling} \right]^{-2}} \\
	&= \sqrt{1+2\coeff{-1}(a+b)\coupling}
	= G_{a+b}(\coupling).
\end{align*}
}%
\begin{beispiel}\label{ex:leading-log}
	The \emph{leading-$\log$} expansion takes only the highest power of $\scalelog$ in each $\coupling$-order. Equally, $\widetilde{\toylog} (\coupling) = c \coupling^n$ for constants $c\in\K$, $n\in\N$ and \eqref{eq:propagator-coupling-ode} integrates to
	\begin{equation}\label{eq:leading-log}
		G_{\text{leading-$\log$}} (\coupling)
		= \Big[ 1 + c n \powdep \scalelog \coupling^n \Big]^{-\frac{1}{n\powdep}}.
	\end{equation}
As a special case we recover example \ref{ex:intrules-correlation} for $n=c=1$ and $\powdep = -2$.
\end{beispiel}
\begin{beispiel}
	In the linear case $\powdep=0$, \eqref{eq:propagator-coupling} states $G_{a+b}(\coupling) = G_a(\coupling) \cdot G_b(\coupling)$ in accordance with the \emph{scaling solution} 
$
	G_{\scalelog}(\coupling) 
	= e^{-\scalelog\tilde{\toylog}(\coupling)}
$
of \eqref{eq:propagator-coupling-ode}, well-known from \cite{Kreimer:linear}.
\end{beispiel}
\begin{beispiel}
	For vertex insertions as in \cite{BKW:NextToLadder} we have $\powdep=1$, so 
$
	G_{a+b}(\coupling) 
	=	G_b(\coupling) \cdot G_a \big[ \widetilde{G}_b(\coupling) \big]
$
expresses the running of the coupling constant $\widetilde{G} \defas \coupling \cdot G$: A change in scale by $b$ is (up to a multiplicative constant) equivalent to replacing the coupling $\coupling$ by $\widetilde{G}_b(\coupling)$.
\end{beispiel}

\subsection{The physicist's renormalization group}

To cast \eqref{eq:propagator-coupling} and \eqref{eq:propagator-coupling-diff} into the common forms of (7.3.15) and (7.3.21) in \cite{Collins}, we introduce the \emph{$\beta$-function} $\beta(\coupling) \defas -\powdep\coupling\widetilde{\toylog}(\coupling)$ and the \emph{running coupling} $\coupling(\rp)$ as the solution of
\begin{equation}\label{eq:running-coupling}
	\rp\frac{\dd}{\dd\rp} \coupling(\rp)
	= \beta \big( \coupling(\rp) \big),
	\quad\text{hence}\quad
	\rp\frac{\dd}{\dd\rp} G \left( \coupling(\rp),\ln\tfrac{s}{\rp} \right)
	\urel{\eqref{eq:propagator-coupling-diff}}
	\widetilde{\toylog}\big( \coupling(\rp) \big) G \left( \coupling(\rp),\ln\tfrac{s}{\rp} \right).
\end{equation}
Integration relates the correlation functions for different renormalization points $\rp$ in
\begin{equation*}
	G\left( \coupling(\rp_2),\ln\tfrac{s}{\rp_2} \right) 
	= G\left( \coupling(\rp_1),\ln\tfrac{s}{\rp_1} \right) \cdot
		\exp \left[ 
			\int_{\rp_1}^{\rp_2} \widetilde{\toylog}\big( \coupling(\rp) \big) \tfrac{\dd\rp}{\rp}
		\right]
	\urel{\eqref{eq:running-coupling}}
		G\left( \coupling(\rp_1),\ln\tfrac{s}{\rp_1} \right) \cdot
		\left[ 
			\frac{\coupling(\rp_2)}{\coupling(\rp_1)}
		\right]^{-\frac{1}{\powdep}}.
\end{equation*}
Setting $\rp_1=s$ we may thus write $G_{\scalelog}(\coupling)$ explicitly in terms of $\widetilde{\toylog}(\coupling)$ as
\begin{equation}\label{eq:rge-solved}
	G_{\scalelog}( \coupling)
	%G\left( \coupling,\ln\tfrac{s}{\rp} \right)
	= \left[
			\frac{\coupling}{\coupling(s)}
		\right]^{-\frac{1}{\powdep}},
	\quad\text{with $\coupling(s)$ subject to}\quad
	\scalelog
	= \ln\frac{s}{\rp} 
	= \int_{\coupling}^{\coupling(s)} \frac{\dd\coupling'}{\beta(\coupling')}.
\end{equation}

\subsection{Relation to Mellin transforms}

We finally exploit the analytic input from theorem \ref{satz:toymodel-universal} to the perturbation series in
\begin{equation*}
	G_{\scalelog}(\coupling)
	\urel{\eqref{eq:dse}}
		1 + \sum_{n\in\N} \coupling^n \toyphy\circ B_n \left(X(\coupling)^{1+n\powdep}\right)
	\urel{\ref{satz:toymodel-universal}}
		1+ \sum_{n\in\N} \coupling^n \left[-\coeff[(n)]{-1}\polyint + P\circ\widetilde{F_n}(\partial_{-\scalelog})\right] G_{\scalelog}(\coupling)^{1+n\powdep},
\end{equation*}
with Mellin transforms
$
	F_n(\reg)
	=\frac{1}{\reg}\coeff[(n)]{-1} + \widetilde{F}_n(\reg)
$
corresponding to the insertions\footnote{For this generality we need decorated rooted trees as commented on in section \ref{sec:decorations}} $B_n$.
\begin{korollar}
	By $G_{\scalelog}(0) = 1$, the power series $G_{\scalelog}(\coupling)\in{\K}[\scalelog][[\coupling]]$ is fully determined by
	\begin{equation}\label{eq:correlation-mellin-de}
		\partial_{-\scalelog} G_{\scalelog} \left( \coupling \right)
		\urel{\eqref{eq:mellin-trafo}} 
		\sum_{n\in\N} \coupling^n \left[\reg F_n(\reg)\right]_{\reg = - \partial_{\scalelog}} \left( {G_{\scalelog}(\coupling)^{1+n\powdep}} \right).
	\end{equation}
\end{korollar}
Restricting to a single cocycle $F_k(\reg) = F(\reg) \delta_{k,n}$, choosing $F(\reg)=\frac{\coeff{-1}}{\reg}$ reproduces \eqref{eq:leading-log} from
$
	\partial_{-\scalelog} G_{\scalelog}(\coupling) 
	= \coupling^n \coeff{-1} G_{\scalelog}(\coupling)^{1+n\powdep}.
$
More generally, for any rational $F(\reg) = \frac{p(\reg)}{q(\reg)} \in \K(\reg)$ with $q(0)=0$, \eqref{eq:correlation-mellin-de} collapses to a finite order ode
$
	q(-\partial_{\scalelog}) G_{\scalelog}(\coupling) 
	= \coupling^n p(-\partial_{\scalelog}) G_{\scalelog}(\coupling)^{1+n\powdep}
$
that makes perfect sense non-perturbatively (extending the algebraic $\partial_{\scalelog}\in\End(\K[\scalelog])$ to the analytic differential operator).

\begin{beispiel}
	For $F(\reg)=\frac{1}{\reg(1-\reg)}$, the propagator ($\powdep=-2$ as in example \ref{ex:DSE-propagator}) fulfils
\begin{equation*}
	\frac{\coupling}{G_{\scalelog}(\coupling)}
	=
	\partial_{-\scalelog} \left( 1-\partial_{-\scalelog} \right) G_{\scalelog} (\coupling)
	\urel{\eqref{eq:propagator-coupling-diff}}
	\widetilde{\toylog}(\coupling) \left(1 -2\coupling\partial_{\coupling} \right)
	\left[ 
	1	-	\widetilde{\toylog}(\coupling) \left(1-2\coupling\partial_{\coupling} \right)
		\right] G_{\scalelog}(\coupling).
\end{equation*}
At $\scalelog=0$ this evaluates to $\widetilde{\toylog}(\coupling) -\widetilde{\toylog}(\coupling) (1-2\coupling\partial_{\coupling})\widetilde{\toylog}(\coupling) = \coupling$, which is studied in \cite{Yeats,Kreimer:ExactDSE}.
\end{beispiel}

\section{Automorphisms of \texorpdfstring{$H_R$}{H\_R}}
\label{sec:H_R}

Applying the universal property to $H_R$ itself, adding coboundaries to $B_+$ leads to
\begin{definition}\label{def:auto}
	For any $\alpha \in H_R'$, theorem \ref{satz:H_R-universal} defines the Hopf algebra morphism
	\begin{equation}
		\auto{\alpha} \defas \unimor{B_+ + \dH\alpha}\!:\ H_R \rightarrow H_R
		\quad\text{such that}\quad
		\auto{\alpha} \circ B_+ = \left[ B_+ + \dH \alpha \right] \circ \auto{\alpha}.
		\label{eq:auto}
	\end{equation}
\end{definition}
\begin{beispiel}\label{ex:auto}
	The action on the simplest trees yields
\begin{align*}
	\auto{\alpha} \left( \tree{+-} \right) 
		&= \auto{\alpha} \circ B_+ (\1)
		 = B_+ (\1) + (\dH\alpha) (\1) 
		 = B_+ (\1)
		 = \tree{+-}, \\
%		 \label{auto:()}\\
	\auto{\alpha} \left( \tree{++--} \right)
		&= \auto{\alpha} \circ B_+ \left( \tree{+-} \right)
		 = \left( B_+ + \dH\alpha \right) \auto{\alpha} \left( \tree{+-} \right)
		 = \tree{++--} + \dH\alpha \left( \tree{+-} \right)
		 = \tree{++--} + \alpha(\1) \tree{+-}, \\
%		 \label{auto:(())}\\
	\auto{\alpha} \left( \tree{+++---} \right)
		&%= \left( B_+ + \dH\alpha \right) \left[ \tree{++--} + \alpha(\1)\tree{+-} \right]
		 = \tree{+++---} + 2 \alpha(\1) \tree{++--} + \left\{ {\left[ \alpha(\1) \right]}^2 + \alpha\left( \tree{+-} \right) \right\} \tree{+-}
%		\label{auto:((()))}
		\quad\text{and}\quad
	\auto{\alpha} \left( \tree{++-+--} \right)
		 %= \left( B_+ + \dH\alpha \right) \left[ \auto{\alpha}(\tree{+-})\auto{\alpha}(\tree{+-}) \right]
		 %= \left[ B_+ + \dH\alpha \right] \left( \tree{+-}\tree{+-} \right)
		 = \tree{++-+--} + 2 \alpha \left( \tree{+-} \right) \tree{+-} + \alpha (\1) \tree{+-}\tree{+-}.
%		 \label{auto:(()())}
\end{align*}
\end{beispiel}
These morphisms capture the change of $\unimor{L}$ under a variation of $L$ by a coboundary in
\begin{satz}\label{satz:change-coboundary-equals-auto}
	Let $H$ denote a bialgebra, $L \in HZ^1_{\counit}(H)$ a 1-cocycle and further $\alpha \in H'$ a functional. Then for $\unimor{L},\unimor{L+\dH \alpha}\!:\ H_R \rightarrow H$ given through theorem \ref{satz:H_R-universal} and $\auto{\alpha \circ \unimor{L}}\!:\ H_R \rightarrow H_R$ from definition \ref{def:auto}, we have
	\begin{equation}
		\unimor{L + \dH \alpha} = \unimor{L} \circ \auto{\left[\alpha\, \circ\, \unimor{L}\right]},
		\quad\text{equivalently}\quad
		\vcenter{\xymatrix{
			{H_R} \ar[r]^{\unimor{L + \dH\alpha}} \ar[d]_{\auto{\alpha \circ \unimor{L}}}  & {H} \\
			{H_R} \ar[ur]_{\unimor{L}} &
			}}
		\quad\text{commutes.}
		\label{eq:change-coboundary-equals-auto}
	\end{equation}
\end{satz}

\begin{proof}
	As both sides of \eqref{eq:change-coboundary-equals-auto} are algebra morphisms, it suffices to prove it inductively for trees: Let it be true for a forest $\f\in\forests$, then it holds as well for the tree $ B_+(\f)$ by
	\begin{align*}
		&\unimor{L} \circ \auto{\left[\alpha \circ \unimor{L}\right]} \circ B_+ (\f)
		\wurel{\eqref{eq:H_R-universal}} \unimor{L} \circ \left[ B_+ + \dH \left( \alpha \circ \unimor{L} \right) \right] \circ \auto{\left[\alpha \circ \unimor{L}\right]} (\f) \\
		&= \Big\{ %\underbrace{\unimor{L} \circ B_+}_{
												L \circ \unimor{L}
											%}
		 					+ %\underbrace{\unimor{L} \circ \left[
								%                         \partial \left( \alpha \circ \unimor{L} \right)
								%												 \right]}_{
																				 (\dH \alpha) \circ \unimor{L}
								%												 }
			 \Big\} \circ \auto{\left[\alpha \circ \unimor{L}\right]} (\f) 
		= \left\{ L + \dH \alpha \right\} \circ 
		 			\underbrace{
							\unimor{L} \circ \auto{\left[\alpha \circ \unimor{L}\right]} (\f)
					}_{
							\unimor{L + \dH\alpha} (\f)
					}
		\urel{\eqref{eq:H_R-universal}} \unimor{L + \dH\alpha} \circ B_+ (\f).
	\end{align*}
	We used $(\dH\alpha) \circ \unimor{L} = \unimor{L} \circ \dH \left( \alpha \circ \unimor{L} \right)$, following from $\unimor{L}$ being a morphism of bialgebras.
\end{proof}

Hence the action of a coboundary $\dH\alpha$ on the universal morphisms induced by $L$ is given by $\auto{\alpha \circ \unimor{L}}$. This turns out to be an automorphism of $H_R$ as shown in

\begin{satz}\label{satz:auto}
	The map $\auto{\cdot}\!:\ H_R' \rightarrow \End_{\text{Hopf}}(H_R)$, taking values in the space of Hopf algebra endomorphisms of $H_R$, fulfils the following properties:
	\begin{enumerate}
		\item For any $\f\in\forests$ and $\alpha\in H_R'$, $\auto{\alpha}(\f)$ differs from $\f$ only by lower order forests:
			\begin{equation}
				\auto{\alpha} (\f) \in \f + H_R^{\abs{\f}-1}
				= \f + \bigoplus_{n=0}^{\abs{\f}-1} H_{R,n}.
				\label{eq:auto-leading-term}
			\end{equation}
		\item $\auto{\cdot}$ maps $H_R'$ into the Hopf algebra automorphisms $\Aut_{\text{Hopf}} (H_R)$.
		Its image is closed under composition, as for any \mbox{$\alpha,\beta \in H_R'$} we have $\auto{\alpha} \circ \auto{\beta} = \auto{\gamma}$ taking
			\begin{equation}
				\gamma = \alpha + \beta \circ {\auto{\alpha}}^{-1}.
				\label{eq:auto-composition}
			\end{equation}
		\item The maps $\dH\!: H_R' \rightarrow HZ^1_{\counit}(H_R)$ and $\auto{\cdot}\!: H_R' \rightarrow \Aut_{\text{Hopf}}(H_R)$ are injective,
		thus the subgroup $\im \auto{\cdot} = \setexp{\auto{\alpha}}{\alpha \in H_R'} \subset \Aut_{\text{Hopf}}(H_R)$ induces a group structure on $H_R'$ with neutral element $0$ and group law $\autoconc$ given by
			\begin{equation}
				\alpha \autoconc \beta 
				\defas \auto{\cdot}^{-1} \left( \auto{\alpha} \circ \auto{\beta} \right)
				\urel{\eqref{eq:auto-composition}}
				\alpha + \beta \circ \auto{\alpha}^{-1}
				\quad\text{and}\quad
				\alpha^{\autoconc -1} = -\alpha \circ \auto{\alpha}.
				\label{eq:auto-group}
			\end{equation}
	\end{enumerate}
\end{satz}

\begin{proof}
	Statement \eqref{eq:auto-leading-term} is an immediate consequence of $\dH\alpha (H_R^n) \subseteq H_R^n$: Starting from $\auto{\alpha}\left( \tree{+-} \right) = \tree{+-}$, suppose inductively \eqref{eq:auto-leading-term} to hold for forests $\f, \f'\in \forests$. Then it obviously also holds for $\f \cdot \f'$ as well and even so for $B_+(\f)$ through
	\begin{equation*}
		\auto{\alpha} \circ B_+(\f)
		=\left[ B_+ + \dH\alpha \right] \circ \auto{\alpha} (\f)
		\subseteq \left[ B_+ + \dH\alpha \right] \left( \f + H_R^{\abs{\f}-1} \right)
		\subseteq B_+(\f) + H_R^{\abs{\f}}.
	\end{equation*}
	This already implies bijectivity of $\auto{\alpha}$, but applying \eqref{eq:change-coboundary-equals-auto} to $L=B_+ + \dH\alpha$ and $\auto{\tilde{\alpha}}$ for $\tilde{\alpha} \defas - \alpha \circ \auto{\alpha}$ shows $\id = \auto{\alpha} \circ \auto{\tilde{\alpha}}$ directly. We deduce bijectivity of all $\auto{\alpha}$ and thus $\auto{\alpha}\in\Aut_{\text{Hopf}}(H_R)$ with the inverse $\auto{\alpha}^{-1}=\auto{\tilde{\alpha}}$. Now \eqref{eq:auto-composition} follows from
	\begin{equation*}
		\auto{\left[\alpha+\beta \,\circ\, \auto{\alpha}^{-1}\right]}
		= \unimor{\left[B_+ + \dH \alpha \right] + \dH \left( \beta \,\circ\, \auto{\alpha}^{-1} \right)}
		\urel{\eqref{eq:change-coboundary-equals-auto}} \unimor{\left[B_+ + \dH \alpha\right]} \circ \auto{\big[\beta \ \circ\ \auto{\alpha}^{-1} \ \circ\ \unimor{\left(B_+ + \dH \alpha\right)}\big]}
		= \auto{\alpha} \circ \auto{\beta}.
	\end{equation*}
	Finally consider $\alpha,\beta \in H_R'$ with $\auto{\alpha} = \auto{\beta}$, then
	$
		0
		= (\auto{\alpha} - \auto{\beta}) \circ B_+
		= \dH \circ (\alpha - \beta) \circ \auto{\alpha}
	$
	reduces the injectivity of $\auto{\cdot}$ to that of $\dH$. But if $\dH\alpha = 0$, for all $n\in\N_0$
	\begin{equation*}
		0
		= \dH\alpha \left( {\tree{+-}}^{n+1} \right)
		= \sum_{i=0}^n \binom{n+1}{i} \alpha \left( {\tree{+-}}^i \right) {\tree{+-}}^{n+1-i}
		\quad\text{implies}\quad
		\alpha\left( {\tree{+-}}^n \right) = 0.
	\end{equation*}
	Given an arbitrary forest $\f\in\forests$ and $n\in\N$, the expression
	\begin{align*}
		0 
		&= \dH\alpha \left( {\tree{+-}}^n \f \right)
		= \f \underbrace{\alpha\left( {\tree{+-}}^n \right)}_0 + \sum_{\f} \sum_{i=0}^n \binom{n}{i}{\tree{+-}}^i \f' \alpha\left( {\tree{+-}}^{n-i} \f'' \right) 
		  + \sum_{i=1}^n \binom{n}{i} \bigg[ {\tree{+-}}^i \f \underbrace{\alpha\left( {\tree{+-}}^{n-i} \right)}_0 + {\tree{+-}}^i \alpha\left( \f {\tree{+-}}^{n-i} \right) \bigg]
	\end{align*}
	simplifies upon projection onto $\K \tree{+-}$ to
	$
		\alpha\left( \f {\tree{+-}}^{n-1} \right)
		= - \frac{1}{n}\sum_{\f:\ \f'=\tree{+-}} \alpha\left( {\tree{+-}}^n \f'' \right)
	$.
	Iterating this formula exhibits $\alpha(\f)$ as a scalar multiple of $\alpha\big( {\tree{+-}}^{\abs{\f}} \big)=0$ and proves $\alpha=0$.
\end{proof}

\subsection{Decorated rooted trees}
\label{sec:decorations}

Our observations generalize straight forwardly to the Hopf algebra $H_R(\decor)$ of rooted trees with decorations drawn from a set $\decor$. In this case, the universal property assigns to each $\decor$-indexed family $L_{\cdot}\!:\ \decor \rightarrow \End(\alg)$ the unique algebra morphism 
	\begin{equation*}
			\unimor{L_{\cdot}}\!: H_R(\decor) \rightarrow \alg
			\quad\text{such that}\quad
			\unimor{L_{\cdot}} \circ B_+^d = L_d \circ \unimor{L_{\cdot}}
			\quad\text{for any $d \in \decor$}.
%			\label{eq:H_R-dekoriert-universal}
	\end{equation*}
For cocycles $\im L_{\cdot} \subseteq HZ^1_{\counit}(\alg)$ this is a morphism of bialgebras and even of Hopf algebras (should $\alg$ be Hopf). For a family $\alpha_{\cdot}\!\!:\ \decor \rightarrow H_R'(\decor)$ of functionals, setting $L_d^{\alpha_{\cdot}} \defas B_+^d + \dH \alpha_d $ yields an automorphism $\auto{\alpha_{\cdot}} \defas \unimor{L_{\cdot}^{\alpha_{\cdot}}}$ of the Hopf algebra $H_R(\decor)$. Theorems \ref{satz:change-coboundary-equals-auto} and \ref{satz:auto} generalize in the obvious way.

In view of the Feynman rules, decorations $d$ denote different graphs into which $B_+^d$ inserts a subdivergence. Hence we gain a family of Mellin transforms $F_{\cdot}$ and theorem \ref{satz:toymodel-universal} generalizes straightforwardly as $\toyphy \circ B_+^d = P\circ F_d(-\partial_{\scalelog}) \circ \toyphy$.

\subsection{Subleading corrections under variations of Mellin transforms}

As an application of \eqref{eq:change-coboundary-equals-auto} consider a change of the Mellin transform $F$ to a different $F'$ that keeps $c_{-1}$ fixed but alters the other coefficients $\coeff{n}$. With $\alpha \defas \toyform' - \toyform$,
	\begin{equation*}
			{\toyphy}' 
			= \unimor{\toycc'}
			= \unimor{\toycc + \dH\alpha}
			= \unimor{\toycc} \circ \auto{\left[ \alpha\, \circ\, \unimor{\toycc} \right]}
			= \toyphy \circ \auto{\left[ \alpha\, \circ\, \toyphy \right]}
	\end{equation*}
translates the new renormalized Feynman rules ${\toyphy}'$ into the original $\toyphy$.
For $\coeff{-1}=-1$, this relates $\toyphy$ to $\intrules=\unimor{\polyint}$ using example \ref{ex:auto} together with $\toyform \circ \intrules (\f) = (-1)^{\abs{\f}} \frac{\abs{\f}!}{\f!} \coeff{\abs{\f}}$ as
\begin{align*}
	\toyphy \left( \tree{+-} \right)
		&= \x
		 = \intrules \left( \tree{+-} \right)
		 = \intrules \circ \auto{\toyform \circ \intrules} \left( \tree{+-} \right)
		,\qquad
	\toyphy \left( \tree{++--} \right)
		 = \frac{\x^2}{2} + \coeff{0} \x
		 = \intrules \left\{ \tree{++--} + \toyform(1) \tree{+-} \right\}
		 = \intrules \circ \auto{\toyform \circ \intrules} \left( \tree{++--} \right)
		 ,\\
	\toyphy \left( \tree{+++---} \right)
		&= \frac{\x^3}{6} + \x^2 \coeff{0} + \x (\coeff[2]{0}-\coeff{1})
		= \intrules \left\{ \tree{+++---} + 2\coeff{0} \tree{++--} + \left[\coeff[2]{0}-\coeff{1} \right] \tree{+-} \right\}
		= \intrules \circ \auto{\toyform \circ \intrules} \left( \tree{+++---} \right)
		\quad\text{and}\\
	\toyphy \left( \tree{++-+--} \right)
		&= \frac{\x^3}{3} + \coeff{0} \cdot \x^2 - 2\coeff{1} \cdot \x
		=\intrules \Big\{ \tree{++-+--} + \coeff{0} \tree{+-}\tree{+-} - 2\coeff{1} \tree{+-} \Big\}
		= \intrules \circ \auto{\toyform \circ \intrules} \left( \tree{++-+--} \right).
\end{align*}
\begin{korollar}\label{satz:change-mellin-auto}
	The new correlation function $\toyphy \circ X = \intrules \circ \widetilde{X}$ equals the original $\intrules$ applied to a modified perturbation series $\widetilde{X}(\coupling)$, fulfilling a Dyson-Schwinger equation differing by coboundaries. By \eqref{eq:auto-leading-term} the leading logs coincide and explicitly
	\begin{align*}
		\widetilde{X} (\coupling)
		\defas \auto{\toyform \circ \intrules} \circ X(\coupling)
		= \1 + 	\sum_{n\in\N} \coupling^n \left( B_n + \dH \toyform_n \right) \left( \widetilde{X}(\coupling)^{1+n\powdep} \right).
	\end{align*}
\end{korollar}

\section{Locality, finiteness and minimal subtraction}
\label{sec:ms}

Consider the regulated but unrenormalized Feynman rules $\toy$. Now setting $\alg\defas\C[\reg^{-1},\reg]]$ and $\phi\defas\toy_1\in\chars{H_R}{\alg}$, \eqref{eq:toymodel-mellin} fixes the scale dependence $\toy_s = \phi \circ \gradAut_{-s\reg}$.
\begin{proposition}\label{satz:finiteness-algebraic}
	For any character $\phi\in\chars{H_R}{\alg}$, the following are equivalent:
	\begin{enumerate}
		\item $\phi^{\convolution -1} \convolution (\phi \circ Y) = \phi \circ (S\convolution Y)$ maps into $\frac{1}{\reg} \C[[\reg]]$, so $\lim\limits_{\reg \rightarrow 0} \phi^{\convolution -1} \convolution (z\phi \circ Y)$ exists.
		\item For every $n\in\N_0$, $\phi^{\convolution -1} \convolution (\phi\circ Y^n) = \phi \circ (S\convolution Y^n)$ maps into $\reg^{-n}\C[[\reg]]$.
		\item For any $s\in\K$, 
			$
				\phi^{\convolution -1} \convolution (\phi\circ\theta_{s\reg}) 
				= \phi\circ (S\convolution \theta_{s\reg})
			$
			maps into $\C[[\reg]]$.
	\end{enumerate}
\end{proposition}
\begin{proof}
	We refer to the accounts in \cite{Manchon,Kreimer:ExactDSE,CK:RH2}, however only $1.\Rightarrow 2.$ is non-trivial and
	\begin{equation*}
		\phi\circ \left( S\convolution Y^{n+1} \right)
		= \phi\circ (S\convolution Y^n) \circ Y 
			+ \left[\phi\circ(S\convolution Y)\right] \convolution \left[ \phi\circ (S\convolution Y^n) \right]
	\end{equation*}
	yields an inductive proof. It exploits 
	$
		(S\circ Y)\convolution\id
		%=(S\convolution \id)\circ Y - S\convolution Y 
		= -S\convolution Y
	$
	in the formula ($\alpha$ arbitrary)
	\begin{equation*}
		S\convolution (\alpha\circ Y) -	(S\convolution\alpha)\circ Y
		= - (S\circ Y)\convolution\alpha
		= - \left[ (S\circ Y)\convolution\id \right]\convolution S \convolution \alpha \nonumber
	  =  S\convolution Y \convolution S \convolution \alpha. \qedhere
	\end{equation*}
\end{proof}
Note that condition 3. is equivalent to the finiteness \ref{satz:finiteness} of the physical limit $\toyphy$ as
\begin{equation*}
	\toyR[s]
	= \toy_{\rp}^{\convolution -1} \convolution \toy_s
	= \phi \circ \left[ (S\circ\gradAut_{-\reg\ln\rp}) \convolution \gradAut_{-\reg s}  \right]
	= \phi \circ (S\convolution \gradAut_{-\reg\ln\frac{s}{\rp}}) \circ \gradAut_{-\reg\ln\rp}.
\end{equation*}
\begin{korollar}
	The anomalous dimension can be obtained from the $\frac{1}{\reg}$-pole coefficients
	\begin{equation}\label{eq:toylog-pole}
		\toylog
		= -\partial_0 \circ \toyphy
		= -\partial_0 \circ \lim_{\reg\rightarrow 0} \phi \circ (S\convolution\gradAut_{-\reg \x})
	%	= \lim_{\reg\rightarrow 0} \phi^{\convolution -1} \convolution (\reg\phi\circ Y)
		= \Res \left[ \phi\circ (S\convolution Y) \right].
	\end{equation}
\end{korollar}

The minimal subtraction scheme $\Rms$ projects onto the pole parts such that $\alg = \alg_- \oplus \alg_+$ where $\alg_- \defas \reg^{-1}[\reg^{-1}]$ and $\alg_+ \defas [[\reg]]$. Though it renders finiteness trivial, its counterterms might depend on the scale $s$ and violate locality. So from \cite{CK:RH2} we need
\begin{definition}
	A Feynman rule $\phi\in\chars{H_R}{\alg}$ is called \emph{local} iff in the minimal subtraction scheme, the counterterm $(\phi\circ\theta_{s\reg})_-$ is independent of $s\in\K$.
\end{definition}
\begin{proposition}
	Locality of $\phi\in\chars{H_R}{\alg}$ is equivalent to the conditions of proposition \ref{satz:finiteness-algebraic}.
\end{proposition}
\begin{proof}
	In case of \ref{satz:finiteness-algebraic},
	$
		\phi \circ \gradAut_{s\reg}
%		= \psi \convolution \psi^{\convolution -1} \convolution \left( \psi \circ \gradAut_{t\reg} \right)
		= (\phi_-)^{\convolution -1} \convolution \big\{
				\phi_+ \convolution
				[ \phi^{\convolution -1} \convolution \left( \phi\circ\gradAut_{s\reg} \right) ]
				\big\}
	$
	is a Birkhoff decomposition by condition 3. such that $\left( \phi \circ \gradAut_{s\reg} \right)_- = \phi_-$ from uniqueness. Conversely, for local $\phi$,
	\begin{equation*}
		0 
		= \Rms \circ \left(\phi \circ \gradAut_{s\reg}\right)_+
		= \Rms \circ \left[
			\left(
				\phi \circ \gradAut_{s\reg}
			\right)_-
			\convolution
				(\phi \circ \gradAut_{s\reg})
			\right]
		=	\Rms \circ \left[
		\phi_- \convolution (\phi \circ \gradAut_{s\reg})
			\right]
	\end{equation*}
	implies
	$
		\C[[\reg]]
		= \ker\Rms
		\supseteq
		\im  \phi_- \convolution \left( \phi \circ \gradAut_{s\reg} \right)
	$
	and convolution with $\phi_+^{\convolution-1} = \phi^{\convolution-1} \convolution \phi_-^{\convolution-1}:\! H_R \rightarrow \C[[\reg]]$ yields condition 3. of \ref{satz:finiteness-algebraic}.
\end{proof}
So we showed algebraically that the problems of finiteness in the {\momscheme} and locality in minimal subtraction are precisely the same. These schemes are related by
\begin{lemma}\label{satz:toyR-MStoyR}
	If $\MStoyR[s]$ denotes the $\Rms$-renormalized Feynman rule, then its scale dependence is given by $\toyphy$ through
$
		\MStoyR
		= \left( \momsch{\rp} \circ \MStoyR \right) \convolution \toyR
$
(as already exploited in \cite{BroadhurstKreimer:Auto}).
\end{lemma}
\begin{proof}
	Locality of the minimal subtraction counterterms $\phi_-$ implies $\momsch{\rp} \circ \phi_- = \phi_-$, hence
	\begin{equation*}
		\left( \momsch{\rp} \circ \MStoyR \right) \convolution \toyR
		= \left[ \momsch{\rp} \circ \left( \phi_- \convolution \toy \right) \right]
		\convolution \left( \momsch{\rp} \circ \toy \right)^{\convolution -1} \convolution \toy
		= \left( \momsch{\rp} \circ \phi_- \right) \convolution \toy
		= \MStoyR. \qedhere
	\end{equation*}
\end{proof}
The physical limit $\ev_{\ln s} \circ \MStoyphy = \lim\limits_{\reg\rightarrow 0} \MStoyR[s]$ yields polynomials $\MStoyphy$ and \ref{satz:toyR-MStoyR} becomes
\begin{korollar}
	The characters $\MStoyphy,\toyphy:\ H_R\rightarrow \K[\x]$ fulfil the relations
	\begin{equation}\label{eq:toyphy-MStoyphy}
		\MStoyphy
		= \left( \counit \circ \MStoyphy \right) \convolution \toyphy %\exp_{\convolution}(-x\toylog).
		,\quad\text{equivalently}\quad
		\Delta \circ \MStoyphy 
		= \left( \MStoyphy \tp \toyphy \right) \circ \Delta.
	\end{equation}
\end{korollar}
In particular, the constant parts $\counit\circ\MStoyphy = \ev_0 \circ \MStoyphy \in \chars{H_R}{\K}$ determine $\MStoyphy$ completely as the scale dependence is governed by $\toyphy$. Using $\toyphy = \exp_{\convolution}(-x\toylog)$, the \emph{$\beta$-functional} 
$
	\MStoyphy
	= \exp_{\convolution} \left( x\beta \right) \convolution \left( \counit \circ \MStoyphy \right)
$
from \cite{CK:RH2} relates to $\toylog$ by conjugation:
	\begin{equation*}
		\beta \convolution \left( \counit \circ \MStoyphy \right)
		= -\left( \counit \circ \MStoyphy \right) \convolution \toylog.
	\end{equation*}
\begin{korollar}
	Applying \eqref{eq:perturbation-convolution-char} to \eqref{eq:toyphy-MStoyphy} expresses the correlation function of the $\Rms$-scheme to the {\momscheme} by a redefinition of the coupling constant:
	\begin{equation*}
		G_{{\scriptscriptstyle\mathrm{MS}},\scalelog}(\coupling)
		= G_{{\scriptscriptstyle\mathrm{MS}},0}(\coupling)	\cdot
		G_{\scalelog} \Big( \coupling\cdot \left[ G_{{\scriptscriptstyle\mathrm{MS}},0}(\coupling) \right]^{\powdep} \Big).
	\end{equation*}
\end{korollar}

\section{Feynman graphs and logarithmic divergences}

In a typical renormalizable scalar \qft, the vertex function is logarithmically divergent and may be renormalized by a simple subtraction as studied above. Referring to \cite{BrownKreimer:AnglesScales} for quadratic divergences, we now restrict to logarithmically divergent graphs with only logarithmic subdivergences, in $D$ dimensions of space-time.

Following the notation established in \cite{BlochKreimer:MixedHodge}, the renormalized amplitude of a graph $\Gamma$ in the Hopf algebra $H$ of Feynman graphs is given by the \emph{forest formula}\footnote{We prefer to work in the \emph{parametric} representation as introduced in \cite[section 6-2-3]{ItzyksonZuber}.}
\begin{equation}\label{eq:forest-formula-parametric}
	\PhiR(\Gamma)
	= \int \Omega_{\Gamma}
		\sum_{F\in\forests(\Gamma)} \frac{(-1)^{\abs{F}}}{\psi_F^{\dimension/2}}
		\ln \frac{
						 \phipsi_{\Gamma/F}	+	\sum\limits_{\scriptscriptstyle\Gamma \neq \gamma\in F} \tilde{\phipsi}_{\gamma/F}
		}{
			\tilde{\phipsi}_{\Gamma/F}+	\sum\limits_{\scriptscriptstyle\Gamma \neq \gamma\in F} \tilde{\phipsi}_{\gamma/F}
		}.
\end{equation}
The forests $\forests(\Gamma)$ account for subdivergences, the first and second \emph{Symanzik polynomials} $\psipol_{\Gamma}, \phipol_{\Gamma}$ depend on the edge variables $\alpha_e$ and we integrate over $\RP_{>0}^{\abs{\Edges{\Gamma}}-1}$ in projective space with canonical volume form $\Omega_{\Gamma}$.
Apart from a scale $s$, $\phipol_{\Gamma}$ depends on dimensionless \emph{angle variables} $\Theta=\set{\frac{m^2}{s}}\cup\set{\frac{p_i \cdot p_j}{s}}$ built from the mass $m$ and external momenta $p_i$. We abbreviate $\phipsi_{\Gamma}\defas \frac{\phipol_{\Gamma}}{\psipol_{\Gamma}}$ and denote evaluation at the renormalization point $(\tilde{s},\tilde{\Theta})$ of the {\momscheme} by a tilde or $\restrict{\cdot}{R}\defas\restrict{\cdot}{(s,\Theta)\mapsto(\tilde{s},\tilde{\Theta})}$.

\begin{definition}\label{def:period}
	Holding the angles $\Theta$ fixed, the \emph{period functional} $\period \in H'$ is given by
\begin{equation}\label{eq:def-period}
	\period(\Gamma)
	\defas \restrict{-\frac{\partial}{\partial \ln s} \PhiR(\Gamma) }{R}
	\quad\text{for any}\quad
	\Gamma \in H.
\end{equation}
\end{definition}
\begin{korollar}
	For any graph $\Gamma\in H$, the value $\period(\Gamma)$ is a period in the sense of \cite{Periods} (provided that $\tilde{s}$ and all $\theta\in\tilde{\Theta}$ are rational) by the formula
	\begin{equation}\label{eq:period-parametric}
		\period (\Gamma)
		\urel{\eqref{eq:forest-formula-parametric}}
		\int \Omega_{\Gamma}
		\sum_{F\in\forests(\Gamma)} \frac{(-1)^{1+\abs{F}}}{\psi_F^{\dimension/2}}
		\frac{
			\tilde{\phipsi}_{\Gamma/F}
		}{
			\tilde{\phipsi}_{\Gamma/F}
			+	\sum\limits_{\scriptscriptstyle\Gamma \neq \gamma\in F}
					\tilde{\phipsi}_{\gamma/F}
		}.
	\end{equation}
\end{korollar}

For primitive (subdivergence free) graphs, \cite{Schnetz:Census} gives equivalent definitions of this period in momentum and position space.
The product rule, \eqref{eq:def-period} and $\restrict{\PhiR}{R}=\counit$ show

\begin{korollar}\label{satz:period-inf-char}
	The period is an infinitesimal character $\period \in \infchars{H}{\K}$ (it vanishes on any graph that is not connected).
\end{korollar}

\subsection{Renormalization group}

\begin{proposition}
	Holding the angles $\Theta$ fixed, differentiation by the scale results in\footnote{%
		This simple form circumvents the decomposition into one-scale graphs utilized in \cite{BrownKreimer:AnglesScales} and therefore holds in the original renormalization Hopf algebra $H$.
	}
	\begin{equation}\label{eq:rge-graphs}
		-\frac{\partial}{\partial \ln s} \PhiR
		= \period \convolution \PhiR.
	\end{equation}
\end{proposition}

\begin{proof}
	Adding $0=\period(\Gamma)-\period(\Gamma)$ and collecting the contributions of $\tilde{\phipsi}_{\gamma/F}$ in $(\ast)$ we find
	\begin{align*}
		-\frac{\partial}{\partial \ln s} \PhiR(\Gamma)
		&
		\urel{\eqref{eq:forest-formula-parametric}}
		\int \Omega_{\Gamma} \left\{
				\frac{1}{\psipol_{\Gamma}^{\dimension/2}} 
				+ \sum_{\set{\Gamma}\neq F \in \forests(\Gamma)} 
							\frac{(-1)^{1+\abs{F}}}{\psipol_F^{\dimension/2}}
							\frac{\phipsi_{\Gamma/F}}{
									\phipsi_{\Gamma/F}
									+ \sum\limits_{\scriptscriptstyle\Gamma\neq\delta\in F} \tilde{\phipsi}_{\delta/F}
							}
			\right\} \\
		&
		\urel{\eqref{eq:period-parametric}}
		\period (\Gamma)
				+ \int \Omega_{\Gamma} \hspace{-6mm}\sum_{\set{\Gamma}\neq F \in \forests(\Gamma)} \frac{(-1)^{1+\abs{F}}}{\psipol_F^{\dimension/2}}
				\frac{
						\left( \phipsi_{\Gamma/F} -\tilde{\phipsi}_{\Gamma/F} \right)
						\sum\limits_{\scriptscriptstyle\Gamma \neq \gamma \in F} \tilde{\phipsi}_{\gamma/F}
				}{
					\Big[
						\phipsi_{\Gamma/F} 
						+ \sum\limits_{\scriptscriptstyle\Gamma \neq \delta \in F} \tilde{\phipsi}_{\delta/F} 
					\Big]
					\cdot
					\Big[
						\tilde{\phipsi}_{\Gamma/F} 
						+ \sum\limits_{\scriptscriptstyle\Gamma \neq \delta \in F} \tilde{\phipsi}_{\delta/F}
					\Big]
				} \\
		&
		\urel{\scalebox{1.5}{$(\ast)$}}
		\period (\Gamma)
		+ \int \Omega_{\Gamma} \hspace{-3mm}
					\sum_{\substack{
						\gamma \subdiv \Gamma\\
						\abs{\comps(\gamma)} = 1
					}}
					\sum_{\gamma \in F \in \forests(\Gamma)}
					\frac{
						(-1)^{1+\abs{F}}
					}{
						\psipol_F^{\dimension/2}
					}
				\frac{
						\left( \phipsi_{\Gamma/F} -\tilde{\phipsi}_{\Gamma/F} \right)
						\tilde{\phipsi}_{\gamma/F}
				}{
					\Big[
						\phipsi_{\Gamma/F} 
						+ \sum\limits_{\scriptscriptstyle\Gamma \neq \delta \in F} \tilde{\phipsi}_{\delta/F}
					\Big]
					\cdot
					\sum\limits_{\scriptscriptstyle \delta \in F} \tilde{\phipsi}_{\delta/F} 
				}.
	\end{align*}
With $\gamma \subdiv \Gamma$ denoting a subdivergence $\gamma\neq \Gamma$, the forests $F\in\forests(\Gamma)$ containing $\gamma$ correspond bijectively to the forests of $\gamma$ and $\Gamma/\gamma$ by
	\begin{align*}
		\forests_{\gamma}(\Gamma)
		&\defas
			\setexp{F \in \forests(\Gamma)}{\gamma\in F}
		\ni F
		\mapsto
		\left( \restrict{F}{\gamma},\ F/\gamma \right)
		\in
		\forests(\gamma) \times \forests(\Gamma/\gamma)
		,\quad\text{using} \\
		\restrict{F}{\gamma}
		&\defas
		\setexp{\delta\in F}{\delta \subdiveq \gamma}
		\quad\text{and}\quad
		F/\gamma
		\defas
		\setexp{\delta/\gamma}{ \delta \in F \quad\text{and}\quad \delta \nosubdiveq \gamma}.
	\end{align*}
	This is an immediate consequence of the definition of a forest, as for $F\in\forests_{\gamma}(\Gamma)$, each $\delta\in F$ is either disjoint to $\gamma$ or strictly containing $\gamma$ (in both cases it is mapped to $\delta/\gamma\in F/\gamma$) or itself a subdivergence of $\gamma$. Thus integrating $\int_0^{\infty} \frac{A-\tilde{A}}{(A+tB)(\tilde{A}+tB)} \dd t = B^{-1} \ln \frac{A}{\tilde{A}}$ in
	\begin{align*}
\hide{		&= \period (\Gamma)
		+ \int \hspace{-3mm}
					\sum_{\substack{
						\gamma \subdiv \Gamma\\
						\abs{\comps(\gamma)} = 1
					}}
					\sum_{\substack{
						F_{\gamma} \in \forests(\gamma) \\
						F \in \forests(\Gamma/\gamma)
					}}
					\Omega_{\gamma} \wedge \Omega_{\Gamma/\gamma}
					\frac{
						(-1)^{1+\abs{F_{\gamma}}+\abs{F}}
					}{
						\psipol_{F_{\gamma}}^{\dimension/2}
						\cdot
						\psipol_F^{\dimension/2}
					}
				\frac{
						\left( \phipsi_{\Gamma/F} -\tilde{\phipsi}_{\Gamma/F} \right)
						\tilde{\phipsi}_{\gamma/F_{\gamma}}
				}{
				\left[ \phipsi_{\Gamma/F} + \sum\limits_{\scriptscriptstyle\Gamma \neq \delta \in F} \tilde{\phipsi}_{\delta/F} \right] \cdot
				\sum\limits_{\scriptscriptstyle \delta \in F} \tilde{\phipsi}_{\delta/F} 
				}\\
}	&= \period (\Gamma)
			+ \int \hspace{-3mm}
					\sum_{\substack{
						\gamma \subdiv \Gamma\\
						\abs{\comps(\gamma)} = 1
					}}
					\Omega_{\gamma} \wedge \Omega_{\Gamma/\gamma}
					\sum_{\substack{
						F_{\gamma} \in \forests(\gamma) \\
						F \in \forests(\Gamma/\gamma)
					}}
					\frac{
						(-1)^{1+\abs{F_{\gamma}}+\abs{F}}
					}{
						\psipol_{F_{\gamma}}^{\dimension/2}
						\cdot
						\psipol_F^{\dimension/2}
					}	 \\
			&\quad\quad\quad\times
			\int\limits_0^{\infty} \frac{\dd t_{\gamma}}{t_{\gamma}}
				\frac{
						\left( \phipsi_{\Gamma/F} -\tilde{\phipsi}_{\Gamma/F} \right)
						\cdot t_{\gamma} \cdot
						\tilde{\phipsi}_{\gamma/{F_{\gamma}}}
				}{
				\Big[
					\phipsi_{\Gamma/F} 
					+ \sum\limits_{\scriptscriptstyle \Gamma \neq \delta \in F} \tilde{\phipsi}_{\delta/F} 
					+ t_{\gamma}\cdot\sum\limits_{\scriptscriptstyle \delta \in F_{\gamma}} \tilde{\phipsi}_{\delta/{F_{\gamma}}} 
				\Big] \cdot
				\Big[
					\sum\limits_{\scriptscriptstyle \delta\in F} \tilde{\phipsi}_{\delta/F} 
					+ t_{\gamma}\cdot\sum\limits_{\scriptscriptstyle \delta \in F_{\gamma}} \tilde{\phipsi}_{\delta/{F_{\gamma}}} 
				\Big]
				} \displaybreak[0]\\
		&= \period (\Gamma)
				+ \int \hspace{-3mm}
					\sum_{\substack{
						\gamma \subdiv \Gamma\\
						\abs{\comps(\gamma)} = 1
					}}
					\hspace{-3mm}
					\Omega_{\gamma}\wedge \Omega_{\Gamma/\gamma}
					\hspace{-3mm}
					\sum_{\substack{
						F_{\gamma} \in \forests(\gamma) \\
						F \in \forests(\Gamma/\gamma)
					}}
					\frac{
						(-1)^{1+\abs{F_{\gamma}}+\abs{F}}
					}{
						\psipol_{F_{\gamma}}^{\dimension/2}
						\cdot
						\psipol_F^{\dimension/2}
					}
				\cdot
				\frac{
					\tilde{\phipsi}_{\gamma/{F_{\gamma}}}
				}{
					\sum\limits_{\scriptscriptstyle \delta\in F_{\gamma}}
					\tilde{\phipsi}_{\delta/{F_{\gamma}}}
				}
				\cdot
				\ln \frac{
					\phipsi_{(\Gamma/\gamma)/F}
					+\sum\limits_{\scriptscriptstyle \delta\in F\setminus\set{\Gamma/\gamma}}
					\tilde{\phipsi}_{\delta/F}
				}{
					\sum\limits_{\scriptscriptstyle \delta \in F}
					\tilde{\phipsi}_{\delta/F}
				}
	\end{align*}
	reduces to the projective $\int \Omega_{\gamma}$ in the edge variables of the subgraph $\gamma$, making use of
	\begin{equation*}
		\abs{F}
		= \abs{\restrict{F}{\gamma}}
			+ \abs{F/\gamma}
		,\quad
		\phipsi_{\delta/F}
		= \begin{cases}
			\phipsi_{(\delta/\gamma)/(F/\gamma)},		&\text{if}\ \gamma \nosupdiveq \delta\in F\\
			\phipsi_{\delta/\restrict{F}{\gamma}},	&\text{if}\ \gamma \supdiveq \delta\in F
		\end{cases}
		\quad\text{and}\quad
		\psipol_F
		= \psipol_{\restrict{F}{\gamma}}
			\cdot \psipol_{F/\gamma}
		.
	\end{equation*}
	The apparent factorization into $\period(\gamma)$ and $\PhiR(\Gamma/\gamma)$ shows that we obtain convergent integrals for each $\gamma\subdiv\Gamma$ individually and may therefore separate into
	\begin{align*}
		&= \period (\Gamma)
				+\hspace{-4mm}
				\sum_{\substack{
					\gamma \subdiv \Gamma\\
					\abs{\comps(\gamma)} = 1
				}}
				\int \Omega_{\gamma}
					\sum_{F_{\gamma} \in \forests(\gamma)}
					\frac{
						(-1)^{1+\abs{F_{\gamma}}}
					}{
						\psipol_{F_{\gamma}}^{\dimension/2}
					}
				\cdot
				\frac{
				\tilde{\phipsi}_{\gamma/{F_{\gamma}}}
				}{
					\sum\limits_{\scriptscriptstyle \delta\in F_{\gamma}}
					\tilde{\phipsi}_{\delta/{F_{\gamma}}}
				}\nonumber\\
		&\quad\quad\quad\times
			\int \Omega_{\Gamma/\gamma}
				\sum_{F \in \forests(\Gamma/\gamma)}
				\frac{
					(-1)^{\abs{F}}
				}{
					\psipol_F^{\dimension/2}
				}
				\cdot
				\ln \frac{
					\phipsi_{(\Gamma/\gamma)/F}
					+\sum\limits_{\scriptscriptstyle \delta\in F\setminus\set{\Gamma/\gamma}}
					\tilde{\phipsi}_{\delta/F}
				}{
					\sum\limits_{\scriptscriptstyle \delta \in F}
					\tilde{\phipsi}_{\delta/F}
				}
		= \period \convolution \PhiR (\Gamma).
	\end{align*}
	Note that the terms $\gamma \tp \Gamma/\gamma$ of $\Delta (\Gamma)$ with $\abs{\comps(\gamma)}>1$ do not contribute here by \ref{satz:period-inf-char}.
\end{proof}

Together with \ref{satz:period-inf-char} and the connected graduation of $H$, this shows
\begin{equation*}
	\PhiR 
	= \sum_{n\in\N_0}
				\frac{(-\scalelog)^n}{n!}
				{\left[
					{\left( - \frac{\partial}{\partial \ln s}\right)}^n \PhiR
					\right]}_{s=\tilde{s}}
	\urel{\eqref{eq:rge-graphs}}
		\sum_{n\in\N_0} 
			\frac{(-\scalelog\cdot\period)^{\convolution n}}{n!} 
			\convolution 
			\restrict{\PhiR}{s=\tilde{s}},
\end{equation*}
where we set $\scalelog \defas \ln \frac{s}{\tilde{s}}$ and the series is pointwise finite. Hence note
\begin{korollar}
	The renormalized Feynman rules
	$
		\PhiR 
		= \restrict{\PhiR}{\Theta=\tilde{\Theta}} \convolution \restrict{\PhiR}{s=\tilde{s}}
	$
	factorize (\cite{BrownKreimer:AnglesScales} gives a different decomposition) into the angle-dependent part $\restrict{\PhiR}{s=\tilde{s}}$ and the scale-dependence $\restrict{\PhiR}{\Theta=\tilde{\Theta}}$ given as the Hopf algebra morphism
	\begin{equation}
		\restrict{\PhiR}{\Theta=\tilde{\Theta}}
		= \exp_{\convolution}\left( -\scalelog \period \right):
		\quad H \rightarrow \K[\scalelog].
	\end{equation}
\end{korollar}

\begin{beispiel}
	For primitive $\Gamma\in\Prim(H)$,
$
	\PhiR(\Gamma)
	= -\scalelog \cdot \period(\Gamma) + \restrict{\PhiR}{s=\tilde{s}}(\Gamma)
$
	disentangles the scale- and angle-dependence. Subdivergences evoke higher powers of $\scalelog$ with angle-dependent factors. \emph{Dunce's cap} of $\phi^4$-theory gives $\period\left( \Graph{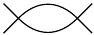} \right) = \period\left( \Graph[0.5]{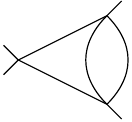} \right) = 1$ such that
\begin{align*}
	\PhiR \left( \Graph{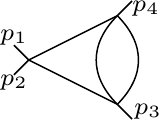} \right)
	&= \frac{\scalelog^2}{2}
		-\scalelog 
		-\scalelog 
			\restrict{\PhiR}{s=\tilde{s}} \left( \Graph{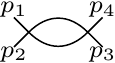} \right)
		+ \restrict{\PhiR}{s=\tilde{s}} \left( \Graph{dunce2_momenta} \right).
\end{align*}
\end{beispiel}

\subsection{Dimensional regularization}

The \emph{dimensional regularization} of \cite{Collins} assigns a Laurent series $\dimPhi(\Gamma)$ in $\reg$ to each Feynman graph $\Gamma\in H$, which for large $\Re \reg$ is given by the convergent parametric integral
\begin{equation}\label{eq:parametric-dimreg}
	\dimPhi(\Gamma)
	= \left[ \prod_{e\in\Edges{\Gamma}} \int\limits_0^{\infty} \alpha_e \right]
	\frac{e^{-\phipsi_{\Gamma}}}{\psipol_{\Gamma}^{\dimension/2 - \reg}}.
\end{equation}
As $\phipsi_{\Gamma}$ is linear in the scale $s$ and homogeneous of degree one in the edge variables, simultaneously rescaling of all $\alpha_e$ yields (for logarithmically divergent graphs)
\begin{korollar}
	The scale dependence
	$
		\dimPhi =
		\restrict{\dimPhi}{s=\tilde{s}} \circ \gradAut_{-\reg\scalelog}
	$
	of \eqref{eq:parametric-dimreg} is induced from the grading $Y$ of $H$ given by the loop number.
\end{korollar}
Thus the finiteness of the physical limit
$
	\restrict{\PhiR}{\Theta=\tilde{\Theta}}
	= \lim_{\reg\rightarrow 0}
		\restrict{\dimPhi}{R} \circ (S \convolution \gradAut_{-\reg\scalelog})
$
results by \ref{satz:finiteness-algebraic} in the local character $\restrict{\dimPhi}{R} \in \chars{H}{\alg}$, evaluated at the renormalization point $(\tilde{s},\tilde{\Theta})$. 
\begin{korollar}
	In dimensional regularization, the period \eqref{eq:period-parametric} is the $\frac{1}{\reg}$-pole coefficient
	\begin{equation}\label{eq:period-dimreg}
		\period 
		\urel{\eqref{eq:toylog-pole}}
		\Res \circ\ \restrict{\dimPhi}{R} \circ (S\convolution Y).
	\end{equation}
\end{korollar}

\subsection{Dilatations and conformal symmetry}

For $\lambda>0$, consider the \emph{dilatation operator} $\dilat_{\lambda}$ scaling masses $m \mapsto \lambda \cdot m$ and momenta $p_i \mapsto \lambda \cdot p_i$.
It fixes all angles $\Theta$, multiplies the scale $s$ with $\lambda^2$ and therefore acts as
\begin{equation*}
	\PhiR \circ \dilat_{\lambda}
	 =	\exp_{\convolution}\left( -\period \ln \frac{s}{\tilde{s}} \right) \convolution \restrict{\PhiR}{s=\tilde{s}} \circ (s\mapsto s\cdot\lambda^2)
	 =	\exp_{\convolution}\left( -2\period \ln\lambda \right) \convolution \PhiR.
	\label{eq:dilatation-action}
\end{equation*}
In other words, the dilatations $\R_{>0}\ni \lambda \mapsto \dilat_{\lambda} \mapsto \exp_{\convolution} \left( -2\period \ln \lambda \right) \convolution \cdot$ are represented on the group $\chars{H}{\alg}$ of characters by a left convolution.
As the unrenormalized logarithmically divergent graphs are dimensionless and naively invariant under $\dilat_{\lambda}$, $\period$ precisely measures how renormalization breaks this symmetry, giving rise to \emph{anomalous dimensions}.

\section{Conclusion}
\label{sec:conclusion}

We stress that the physical limit of the renormalized Feynman rules results in a morphism \mbox{$\toyphy\!: H_R \rightarrow \K[x]$} of Hopf algebras in case of the {\momscheme}. This compatibility with the coproduct allows to obtain $\toyphy$ from the linear terms $\toylog$ only. As we just exemplified, these relations are statements about individual Feynman graphs unraveling scale- and angle-dependence in a simple way. Again we recommend \cite{BrownKreimer:AnglesScales} for further reading.

Secondly we revealed how Hochschild cohomology governs not only the perturbation series through Dyson-Schwinger equations, but also determines the Feynman rules. Addition of exact one-cocycles captures variations of Feynman rules and the anomalous dimension $\toylog$ can efficiently be calculated in terms of Mellin transform coefficients.

Note how this feature is lost upon substitution of the {\momscheme} by minimal subtraction: We do not obtain a Hopf algebra morphism anymore due to the constant terms, which are also more difficult to obtain in terms of the Mellin transforms $F$.

Finally we want to emphasize the remarks in section \ref{sec:DSE} towards a non-perturbative framework. Though this relation between $F(\reg)$ and the anomalous dimension $\widetilde{\toylog}(\coupling)$ is still under investigation and so far only fully understood in special cases, these already give interesting results \cite{Kreimer:ExactDSE,Yeats}.

\bibliographystyle{plainurl}
\bibliography{qft}

\end{document}